\documentclass{amsart}
\usepackage{graphicx}
\usepackage{color}
\usepackage{mathrsfs}
\usepackage{amsthm}
\newtheorem{theorem}{Theorem}[section]

\newtheorem{remark}{Remark}[subsection]
\newtheorem{lemma}{Lemma}[section]
\newtheorem{cor}{Corollary}[section]

\title{On the three-dimensional stability of Poiseuille flow in a finite-length duct}
\author[*]{Lei Xu\textsuperscript{1} \and Zvi Rusak\textsuperscript{1}}\thanks{Department of Mechanical, Aerospace and Nuclear Engineering, Rensselaer Polytechnic Institute}

\begin{document}
\maketitle

\begin{abstract}
	The stability of a three-dimensional, incompressible, viscous flow through a finite-length duct is studied. A divergence-free basis technique is used to formulate the weak form of the problem. A SUPG(streamingline upwind Petrov-Galerkin) based scheme for eigenvalue problems is proposed to stabilize the solution. With proper boundary condtions, the least-stable eigenmodes and decay rates are computed. It is again found that the flows are asymptotically stable for all $Re$ up to $2500$. It is discovered that the least-stable eigenmodes have a boundary-layer-structure at high $Re$, although the Poiseuille base flow does not exhibits such structure. At these Reynolds numbers, the eigenmodes are dominant in the vicinity of the duct wall and are convected downstream. The boundary-layer-structure brings singularity to the modes at high $Re$ with unbounded perturbation gradient. It is shown that due to the singular structure of the least-stable eigenmodes, the linear Navier-Stoker operator tends to have pseudospectrua and the nonlinear mechanism kicks in when the perturbation energy is still small at high $Re$. The decreasing stable region as $Re$ increases is a result of both the decreasing decay rate and the singular structure of the least-stable modes. The result demonstrated that at very high $Re$, linearization of Navier-Stokes equation for duct flow may not be a good model problem with physical disturbances.
\end{abstract}

\section{Introduction}\label{S:intro}
The stability of pipe or duct Poiseuille flow and its transition to turbulence remains an active topic. Reynolds~\cite{reynolds} first demonstrated that the pipe flow usually makes a transistion to turbulence at Reynolds number near $2000$. It was shown by Huerre and Rossi~\cite{heurre} that laminar flow can be achieved up to Reynolds number $~10^5$ by carefully reducing external disturbances. Various numerical simulations have been studied to analyze pipe flow stability, including~\cite{boberg},~\cite{sullivan},~\cite{priymak},~\cite{zikanov}. In addition, the linearized pipe Poiseuille flow has been investigated by~\cite{bergstrom},~\cite{schmid},~\cite{trefethen},~\cite{meseguer}. Duct flows with various cross-section shapes were studied by Delplace~\cite{delplace} and similar behavior as pipe flow was found. It seems such flows are linearly stable for all Reynolds numbers, yet they tend to be unstable above certain Reynolds number in experiments, see~\cite{darbyshire},~\cite{wygnanski}. 

One way to explain such phenomena is to analyze the nonlinear dynamics and identify the disturbance threshold that can keep the laminar flow. The inverse relationship between the level of flow perturbation and the Reynolds number up to which the laminar flow can sustain has been studied by~\cite{hof},~\cite{trefethenscience}. A question is therefore raised to identify the scaling factor $\gamma$, where $\epsilon=O(Re^{\gamma})$, with $\epsilon$ the mimial amplitude of all perturbations that can trigger instability. Hof~\cite{hof} showed by experiments that the factor $\gamma=-1$, while other estimates ranges from $-1$ to $-7/4$, see~\cite{meseguernonlin},~\cite{eckhardt}. Another way to understand this sensitivity is to look at the pseudospectra of the linearized problem, see also~\cite{meseguer},~\cite{trefethenscience}. The linear operator resulting from Navier-Stokes equations exhibits strong sensitivity at high $Re$ with respect to pipe geometry fluctuations and thus instability can occur with finite $Re$.

In this paper we computed the stability of three-dimensional Poiseuille flow in a finite-length duct, with physical boundary conditions that fixes the inlet velocity and the normal stresses at the outlet. The three-dimensional structure of the least-stable modes are resolved and the eigenvalues are computed. It is again shown that the duct Poiseuille flow is linearly stable up to $Re=2500$. The study uses a SUPG(Streamline upwind Petrov-Galerkin) based finite element method to compute a generalized matrix eigenvalue problem. The SUPG scheme, proposed by Hughes~\cite{hughes1979},~\cite{hughes1982},~\cite{hughes1987}, is well-known to provide the answer of finding a higher-order accurate method for the advection-diffusion equations without deterioration of convergence rate.  The numerical scheme also employs a weakly-divergence-free basis technique to overcome the singular mass matrix difficulty brought by the incompressibility condition. The divergence-free basis numercial scheme has been studied by various articles. Griffths~\cite{griffiths1979},~\cite{griffiths1979-2},~\cite{griffiths1981} proposed several element-based divergence-free basis on triangular and quadrilateral elements. Fortin~\cite{fortinnewfem} also investigated the discrete divergence-free subspace with piecewise constant pressure space and various discrete velocity spaces. Ye and Hall~\cite{ye} showed another discrete divergence-free basis with continuous discrete pressure space. The main advantage of weakly-divergence-free basis technique is that the degrees of freedom is greatly reduced when solving a large scale matrix problem. Apparantly, this method gives added benefits when analyzing the eigenvalue problem with incompressibility constraint. 

The computation result shows a boundary-layer-structure of the least-stable eigenmodes at high $Re$. As $Re$ increases, the modes are dominant close to the duct wall, forming a thin boundary layer with large gradients. The interesting structure, however, is not accompanied by a boundary layer of the base flow. The unbounded mode gradient is crucial in understanding the nonlinear effects generated by small perturbations. In high $Re$ flows, even physically small disturbances can induce non-negligible nonlinear effects, due to the large value of $\mathbf{u}\cdot\nabla \mathbf{u}$ where $\mathbf{u}$ denotes the velocity perturbation. The boundary-layer-structure also qualitatively explains the sensitivity of the linearized Navier-Stokes problem in the pseudospectra theory. Due to the eigenmodes dominance in the vicinity of the duct wall,  small fluctuations of duct geometry may significantly alter the least-stable mode shapes. This may cause large deviation of eigenvalues from ideal with small linear operator perturbations. This result gives another point of view regarding studies in~\cite{meseguer},~\cite{trefethenscience}. 

%Finally, the nonlinear stability threshold is analyzed at various $Re$ from the results of the computed eigenvalues and eigenmodes. 
%The nonlinear analysis assumes initial condition to be the least-stable mode and then the critical amplitude of the flow perturbation at which the flow energy starts to grow is computed. This assumption is different from the other studies to induce impulsive or periodic distrubances. The boundary-layer shaped intial perturbation is not very likely to be physically induced. However, if we assume the linearization process is valid, the disturbance will reshape into the least-stable modes, thus analyzing the nonlinear effects of the modes are meaningful. If an eigenmode with small amplitude causes the flow to be unstable when its nonlinear effect is considered, the linearization assumption is no longer valid, by proof of contradiction. Studying the maximum amplitude of eigenmodes that validates the linearization assumption sheds light on the understanding of Poiseuille flow stability.

\section{Mathematical Model}\label{S:MathModel}
\subsection{Unsteady Navier-Stokes equations}
We consider a visous, incompressible Newtonian flow through a three-dimensional duct. The duct has a sqaure cross-section. The distances are scaled with the duct width $d$ and the non-dimensionalized duct length is $L$, see figure~\ref{F:Ductgeometry}. Rectangular coordinates are used, with $x$-$y$ plane lying on the duct cross-section and $z$ axis extending along the duct length. The flow domain is given by $\Omega=\{(x,y,z)|0\leq x\leq 1,0\leq y\leq 1,0\leq z\leq L\}$. The velocity components are scaled with the characteristic speed $U$, such that the nondimensionalized flux through the duct is unity. The fluid density $\rho$  and the viscosity $\mu$ are constants. The pressure is scaled with the dynamic pressure $\rho U^2$. The time is scaled with $\frac{d}{U}$. As a result, the nondimensional viscosity is  $\nu=\frac{1}{Re}=\frac{\mu}{\rho UL}$. The nondimensionalized Navier-Stokes equations are described as follows,
\begin{figure}
	\centering
	\includegraphics[width=5.22in]{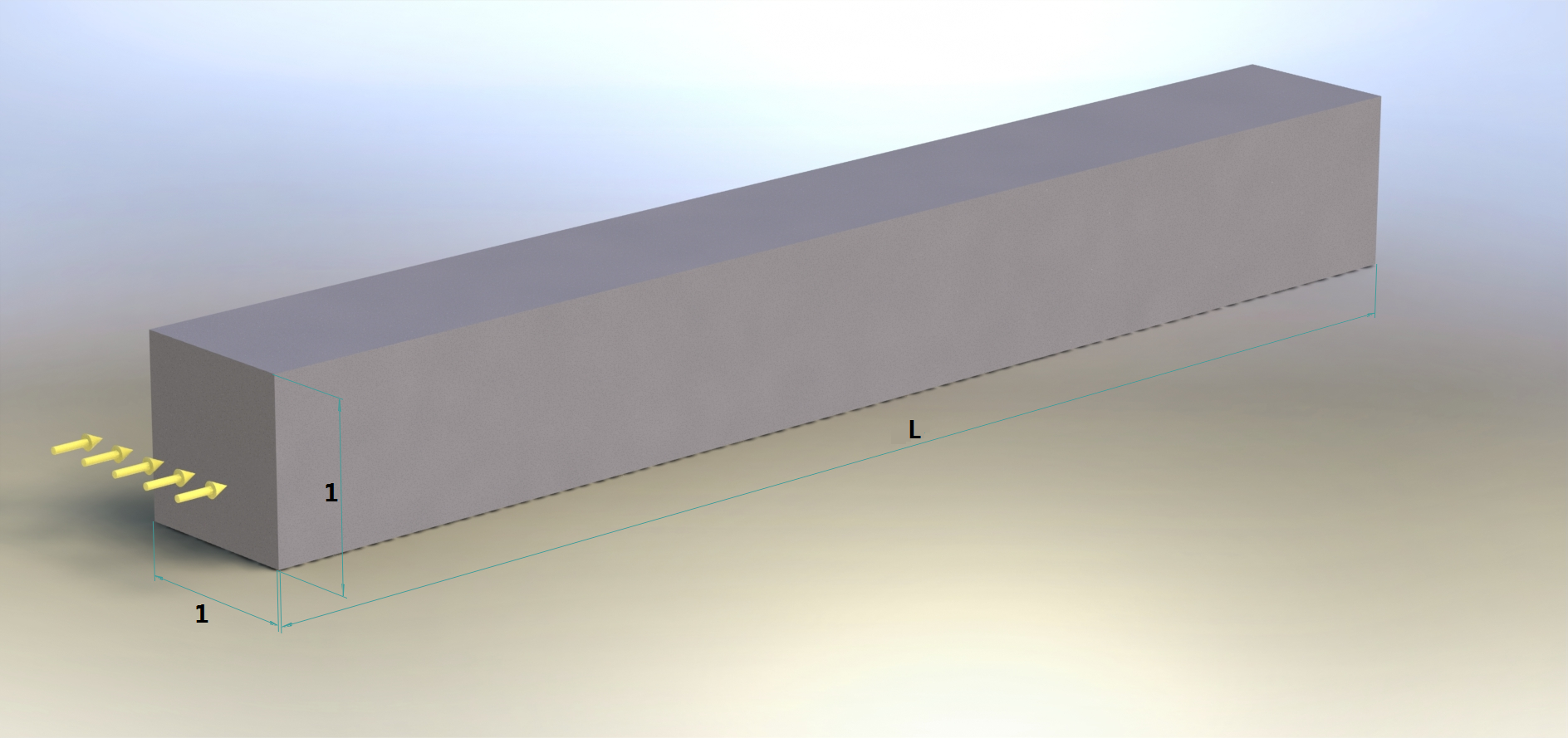}
	\caption{Duct geometry}
	\label{F:Ductgeometry}
\end{figure}

\begin{align}
  \frac{\partial{\mathbf{u}}}{\partial{t}} + \mathbf{u} \cdot \nabla \mathbf{u} &= -\nabla p + \nu \nabla^2\mathbf{u},~ in~ \Omega \label{E:unsteadyNS}\\
  \nonumber \nabla \cdot \mathbf{u} &= 0,~ in~ \Omega,
\end{align}
with boundary conditions,
\begin{align}
  inlet ~condition&:~ \mathbf{u} = \mathbf{u}_0(x,y),~ z=0,~for~all~t\ge 0\\
  wall ~conidtion&:~  \mathbf{u} = \mathbf{0}, ~ x=0, x=1, y=0, y=1,~for~all~t\ge 0 \\
  outlet ~condtion&:~ -p\mathbf{n} + \nu \frac{\partial{\mathbf{u}}}{\partial{\mathbf{n}}} = 0,~ z=L,~for~all~t\ge0,
\end{align}
and initial condition is,
\begin{align}
	\nonumber \mathbf{u} = \mathbf{f}_0(x,y,z),~when~ t=0, ~\mathbf{f}_0 ~satisfies~ above ~B.C..
\end{align}
Here the inlet velocity is fixed to the inlet velocity profile $\mathbf{u}_0(x,y)$, which will be described in detail shortly after. The duct wall imposes a no-slip condition. The outlet is subjected to a natural Neumann condition. The vector $\mathbf{n}$ denotes the unit outward normal vector of the duct outlet at $z=L$. This outlet condition decouples the velocity field $\mathbf{u}$ and the pressure field $p$ in the weak form. At high Reynolds number, the outlet boundary condition approaches a constant pressure outlet condition. The initial condition can be chosen arbitrarily.

It is known that there exists a steady state solution $\mathbf{u}_0(x,y,z)=[0,0,w_0(x,y)]^T$, $p(x,y,z)=p_0(z)$ to the above model problem, satisfying,
\begin{align}
	\frac{\partial^2 w_0}{\partial{x^2}}+\frac{\partial^2 w_0}{\partial{y^2}} =\frac{1}{\nu} \frac{\partial p_0}{\partial z}. \label{E:steadystate}
\end{align}
Here $\frac{1}{\nu} \frac{\partial p_0}{\partial z}$ is a negative scaled constant such that the constraint of total flux $1$ is satisfied. It can be verified that the solution of eqation~(\ref{E:steadystate}) is,
%\begin{aligned}
	%w_0(x,y) = &K \left[ -\frac{1}{4}(y^2-y)-\frac{1}{4}(x^2-x) \right] \\
%\nonumber	&+ K \left[ \frac{1}{2}\sum_{m=1}^{\infty}sin(m\pi x)(A_m e^{m \pi y}+B_m e^{-m \pi y}) + \frac{1}{2}\sum_{m=1}^{\infty}sin(m\pi y)(C_m e^{m \pi x}+D_m e^{-m \pi x} ) \right],   \label{E:steadystatew0} \\
%\end{aligned}
\begin{equation}
\begin{aligned}
w_0(x,y) = &K_0  \left[ -\frac{1}{4}(y^2-y)-\frac{1}{4}(x^2-x) + \right. \\
&\left. \frac{1}{2}\sum_{m=1}^{\infty}sin(m\pi x)(A_m e^{m \pi y}+B_m e^{-m \pi y}) +\right. \\
 &\left. \frac{1}{2}\sum_{m=1}^{\infty}sin(m\pi y)(C_m e^{m \pi x}+D_m e^{-m \pi x} ) \right], \\
 &0\leq x\leq 1, ~ 0\leq y\leq 1
  \label{E:steadystatew0}
\end{aligned}
\end{equation}
\begin{align}
	\nonumber A_m &=-\frac{2}{(m\pi)^3}\left[ (-1)^m-1\right] \frac{e^{-m\pi}-1}{e^{m \pi}-e^{-m\pi}}, \\
	\nonumber B_m &=\frac{2}{(m\pi)^3}\left[ (-1)^m-1\right] \frac{e^{m\pi}-1}{e^{m \pi}-e^{-m\pi}}, \\
	\nonumber C_m &=A_m, \\
	\nonumber D_m &=B_m, \\
	\nonumber K_0 &= 28.4541538.
\end{align}

Given the steady state solution~(\ref{E:steadystatew0}), we consider the associated linear stability problem. We denote the perturbation to the base flow $\mathbf{u}_0$ as $\mathbf{u'_\epsilon}$. The velocity field can be expressed as $\mathbf{u}=\mathbf{u}_0+\mathbf{u'_\epsilon}$.  Similarly the pressure field $p$ can be expressed as $p=p_0+p'_\epsilon$. By neglecting higher oder terms when the perturbations are small, the linearized Navier-Stokes equations are,
\begin{align}
	 \frac{\partial{\mathbf{u'_\epsilon}}}{\partial{t}} + \mathbf{u_0} \cdot \nabla \mathbf{u'_\epsilon} +  \mathbf{u'_\epsilon} \cdot \nabla \mathbf{u_0} &= -\nabla p'_\epsilon + \nu \nabla^2\mathbf{u'_\epsilon},~ in~ \Omega\\
	\nonumber  \nabla \cdot \mathbf{u'_\epsilon} &= 0,~ in~ \Omega.
\end{align}	
 Boundary conditions are,
\begin{align}
	inlet ~condition&:~ \mathbf{u'_\epsilon} = \mathbf{0},~ z=0,~for~all~t\ge 0  \label{E:eigenBCinlet} \\
	wall ~conidtion&:~  \mathbf{u'_\epsilon} = \mathbf{0},~  x=0, x=1, y=0, y=1,~for~all~t\ge 0 \label{E:eigenBCwall} \\
	outlet ~condtion&:~ -p'_\epsilon\mathbf{n} + \nu \frac{\partial{\mathbf{u'_\epsilon}}}{\partial{\mathbf{n}}} = 0, ~z=L,~for~all~t\ge0, \label{E:eigenBCoutlet}
\end{align}
and initial condition is,
\begin{align}
	\nonumber \mathbf{u'_\epsilon} = \mathbf{f'}_{\epsilon 0}(x,y,z),~when~ t=0,~\mathbf{f}_{\epsilon 0} ~satisfies~ above ~B.C..
\end{align}

\subsection{Eigenvalue problem formulation}
We now turn to study the spectral decomposition of the above linear system. The mode can be expressed as,
\begin{align}
\mathbf{u'_\epsilon} &= e^{\lambda t}\mathbf{u'}(x,y,z), \\
\nonumber p'_\epsilon &= e^{\lambda t}p'(x,y,z).
\end{align}
 For simplicity, we will omit the prime $'$ of the velocity perturbation mode $\mathbf{u'}$ and the pressure perturbation mode $p'$ from now on.
 The equations of the corresponding eigenvalue problem after substituting the velocity and pressure fields are,
 \begin{align}
 	\lambda \mathbf{u} = -\mathbf{u_0} \cdot \nabla \mathbf{u} & -\mathbf{u} \cdot \nabla \mathbf{u_0}  -\nabla p + \nu \nabla^2 \mathbf{u}  \label{E:eigenNS1}   \\
 \nonumber	\nabla \cdot \mathbf{u} &= 0 
 \end{align}
 Here $\mathbf{u}$ and $p$ are mode shapes and they are functions of space coordinates $(x,y,z)$ only.
 Equations~(\ref{E:eigenNS1}) are subjected to boundary conditions~(\ref{E:eigenBCinlet}),~(\ref{E:eigenBCwall}),~(\ref{E:eigenBCoutlet}).
 
 The sign of eigenvalue $\lambda$ determines whether each mode is asymptotically stable or not. It is crucial to find the least-stable eigenvalues with the greatest real part. These most dangerous modes are dominant either when the flow perturbation is decaying, or when the flow is unstable as the perturbation grows exponentially.
 
 \subsection{Divergence-free subspace and weak form}
 The classical approach in finding the solutions of equations~(\ref{E:eigenNS1}) is to solve the velocity mode $\mathbf{u}$ in functional space $\mathbf{V}=\left\lbrace  \mathbf{u}\in [H^1(\Omega)]^3 |~ \mathbf{u} ~satisfies ~boundary~ condition~(\ref{E:eigenBCinlet})-(\ref{E:eigenBCoutlet}) \right\rbrace $ and the pressure mode in $Q=\left\lbrace p| p\in L^2(\Omega)/\mathbb{R} \right\rbrace $. We now try to find the solution of $\mathbf{u}$ in its subspace  $\mathbf{V}^0 = \left\lbrace  \mathbf{u} \in \mathbf{V} | \nabla \cdot \mathbf{u}=0 \right\rbrace $. When limiting the solution space in the subspace $\mathbf{V}^0$, the incompressiblity constraint is no longer needed. Thus a weak form can be expressed from equation~(\ref{E:eigenNS1}),
 \begin{align}
 	\nonumber find~&\mathbf{u} \in\mathbf{V}^0,~\lambda\in\mathbb{C}~s.t, \label{E:weakform1}\\
 	\lambda M_0(\mathbf{u},\mathbf{w}) &= A_0(\mathbf{u},\mathbf{w})+C_0(\mathbf{u},\mathbf{w}),~\forall~\mathbf{w}\in \mathbf{V}^0, \\
 	\nonumber M_0(\mathbf{u},\mathbf{w}) &= \left<u^i,w^i \right>, \\
 	\nonumber A_0(\mathbf{u},\mathbf{w}) &= -\nu \left<\nabla_k u^i, \nabla_k w^i \right>, \\
 	\nonumber C_0(\mathbf{u},\mathbf{w}) &= \left< -u_0^k \nabla_k u^i - u^k \nabla_k u_0^i, w^i \right>, \\
 	\nonumber ||\mathbf{u}||_{L^2(\Omega)} &= 1. 
 \end{align}
 The inner-product is defined as the integral over $\Omega$, i.e., $\left\langle f,g \right\rangle = \int_{\Omega}fgd\Omega$.
 Note the pressure term and boundary integrals are eliminated as a result of the specified boundary conditions~~(\ref{E:eigenBCinlet})-(\ref{E:eigenBCoutlet}). The eigenfunctions are normalized such that its $L_2$ norm is unity.
 
   \subsection{Finite element formulations}
   A $Q_1-P_0$ mixed form is used in this paper. It is well-known that the trilinear velocity-constant pressure element suffers from a checkerboard pressure mode on regular meshes. The spurious pressure mode is due to the fact that this element does not satisfy the Bab\v{u}ska-Brezzi(B.B.) condition. However, if the velocity field is the only interested variable, the spurious pressure mode is irrelevant, since pressure term is eliminated in equation~(\ref{E:eigenNS1}). Boffi \emph{et al.}~(\cite{boffietal1997}) have shown that such mixed form is valid in analyzing eigenvalue problem with the Galerkin method.
   
   The finite element formulation first constructs a weakly-divergence-free basis with $Q_1$ elements. Then a discretized weak form from equation~(\ref{E:weakform1}) is formulated. 
   
   Let $\mathbf{V}_h$ and $Q_h$ be the finite dimensional subspace of $\mathbf{V}$ and $Q$.
   The weakly divergence-free subspace follows as,
   \begin{align}
   \mathbf{V}_h^0 &= \left\lbrace \mathbf{u}_h \in \mathbf{V}_h |	\left\langle \nabla \cdot \mathbf{u}_h, q_h \right\rangle =0, ~\forall q_h \in Q_h \right\rbrace \label{E:weaklydivfree} \\
   \nonumber &=Span(\mathbf\Phi_1,\dots,\mathbf\Phi_N)
   \end{align}
   Here the divergence-free basis function $\mathbf\Phi_i$ belongs to the $Q_1$ trilinear element. The trial function $q_h$ is constant in each element. The discrete subspace  is equivalent as a subspace imposing the conservation of mass condition on each element $\bar \Omega^e$, where $\Omega^e$ denotes the interior of the eth element.
   \begin{align}
   	\int_{\bar \Omega^e} \nabla \cdot \mathbf{u}_h d\Omega = \int_{\partial{\bar \Omega^e}} \mathbf{u}_h\cdot \mathbf{n} d\sigma =\sum_{i=1}^{6}\int_{F_i} \mathbf{u}_h \cdot \mathbf{n} d\sigma=  0, \forall \bar \Omega^e,
   \end{align}
   	where $F_i$ are the six faces of a cubic element.
   Fortin~\cite{fortinnewfem} showed that in three-dimensional cubic regular element, the weakly-divergence-free basis $\mathbf\Phi_i$ can be expressed as vortices lying on each face of the element, shown in figure~\ref{F:vortex}. The six faces of a cubic element is associated with six localized vortices. The basis can be physically viewed as a basis for a vorticity field $\mathbf{\omega}=\nabla \times \mathbf{u}$, in the discretized sense. The direction of the 'effective vorticity' is normal to each element face. Note that for each element $\Omega^e$, the six basis functions are linearly dependent and the degree of freedom is five.
   
   \begin{figure}
   	\centering
   	\includegraphics[width=3.22in]{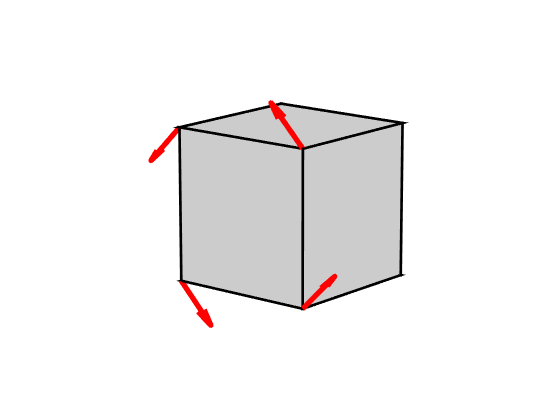}
   	\caption{Weakly-divergence-free basis}
   	\label{F:vortex}
   \end{figure}
   
   \begin{figure}
   	\centering
   	\includegraphics[width=3.22in]{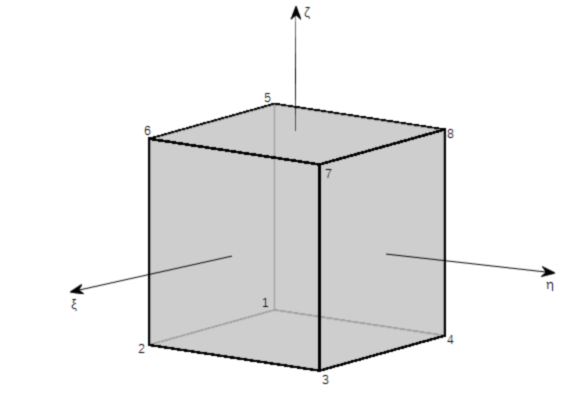}
   	\caption{Reference element}
   	\label{F:referencecube}
   \end{figure}
   
   One of the weakly-divergence-free basis can be expressed in the reference domain as,
   \begin{align}
   	\mathbf{\hat \Phi}_1(\xi,\eta,\zeta) =[0, \hat\phi_2+\hat\phi_3-\hat\phi_6-\hat\phi_7 , -\hat\phi_2+\hat\phi_3-\hat\phi_6+\hat\phi_7]^T,
   \end{align}
   where $\hat\phi_i(\xi,\eta,\zeta)$ is the typical trilinear $Q_1$ shape function, with nodal value $1$ on each of the element vertex. The reference element domain is shown in figure~\ref{F:referencecube}, the domain $\hat T = \left\lbrace (\xi,\eta,\zeta) | ~\xi\in[-1,1],\eta\in[-1,1],\zeta\in[-1,1] \right\rbrace $. Shape function index is associated with node index described in figure~\ref{F:referencecube}. The other weakly-divergence-free basis $\mathbf{\hat \Phi}_i(\xi,\eta,\zeta)$ on the reference domain can be constructed similarly.

   The finite element uses a simplest uniform cubic element discretizing domain $\Omega$ throughout this paper. The cube has identical edge length $h$ in each direction. As a result, the coordinate transformation between $(x,y,z)$ and $(\xi,\eta,\zeta)$ has a simple linear relationship,
   \begin{equation}
   \begin{cases}
   \nonumber x = \frac{h}{2} \xi + x_0, \\
   \nonumber y = \frac{h}{2} \eta + y_0, \\
   \nonumber z = \frac{h}{2} \zeta + z_0, \\
   \end{cases}
   \end{equation}
   where $x_0,y_0,z_0$ are coordinates of the element center. With the above assumptions, it can be verified that the weakly-divergence-free condition is satified for each $\mathbf{\Phi_i}(x,y,z)=\mathbf{\hat \Phi}_i(\xi,\eta,\zeta)$,
    \begin{align}
    	\int_{\partial{\bar \Omega^e}} \mathbf{\Phi_i} \cdot \mathbf{n} d\sigma =0, \forall \bar \Omega^e.
    \end{align}

    It should be pointed out that due to the non-zero outlet boundary condition~(\ref{E:eigenBCoutlet}), 'half vortex' elements must be included in addition to the 'full vortex' element described above, see figure~\ref{F:halfvortex}. One of the 'half vortex' element on the reference domain is, 
       \begin{align}
       	\mathbf{\hat \Phi}_7(\xi,\eta,\zeta) =[0, \hat\phi_3-\hat\phi_7 , -\hat\phi_3-\hat\phi_7]^T,
       \end{align}
       The index $7$ is labelled after the $6$ 'full vortices' on each of the element faces. The other basis functions at the duct outlet can be expressed in a similar way. The transformation to physical domain is straightforward. It can be verified that the 'half vortex' basis satisfies the weakly-divergence-free condition. There is also a linear dependency of four 'half vortices' $\mathbf{\hat \Phi}_7$ to $\mathbf{\hat \Phi}_{10}$ and a 'full vortex' $\mathbf{\hat \Phi}_6$ on the element outlet surface. Thus global basis $\mathbf{\Phi}_i$ mapped from $\mathbf{\hat \Phi}_6$ must be excluded from the final basis set.
    
    \begin{figure}
        	\centering
        	\includegraphics[width=3.22in]{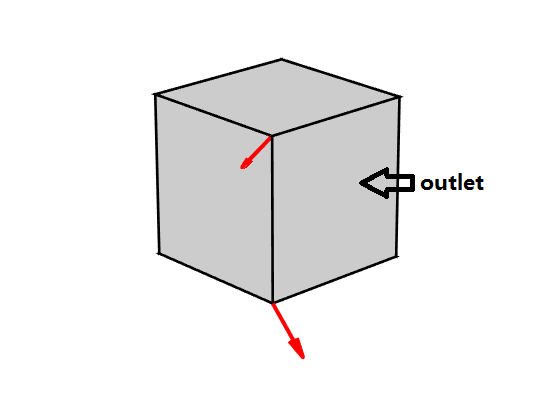}
        	\caption{Weakly-divergence-free basis at outlet}
        	\label{F:halfvortex}
    \end{figure}
   
   With $N_xN_yN_z$ elements in the flow domain, where $N_x$, $N_y$, $N_z$ are the number of elements in $x$, $y$, $z$ directions, the total degree of freedom in the finite element formulation is,
   \begin{align}
   N_{dof}=&(N_x-2)(N_y-1)(N_z-2)+(N_x-1)(N_y-2)(N_z-2) \\
    \nonumber         &+2(N_x-2)(N_y-1)+2(N_x-1)(N_y-2).
   \end{align}
   This is a result of the linear dependency of the six basis functions on each element and the inclusion of boundary vortex elements. It should be noted that all vortices with 'effective vorticity' along the $z$ direction are linearly dependent with the rest of the vortices, thus they must be excluded from the divergence-free basis set. The basis functions $\mathbf\Phi_i$ are constructed such that the 'effective vorticity' has positive component in each direction. This ensures the consistency when $\mathbf\Phi_i$ is evaluated on nearby elements.
   
   Although the construction of the finite element scheme enforces stric cubic elements, it is perhaps the simplest way to analyze the problem with minimal degrees of freedom. 
   
  \subsection{Streamline upwind Petrov-Galerkin(SUPG) method}
  Hughes~\cite{hughes1987} introduced SUPG scheme for computing time-dependent and steady state fluid problems, and enjoyed great success. When element Peclet number, defined as $\alpha=h|\mathbf{u}_0|/(2\nu)$, where $|\mathbf{u}_0|$ is the maximum magnitude of $\mathbf{u}_0$ in each element, is large, the Galerkin method suffers from wild oscillations that contaminate the numerical soltuions. In eigenvalue analysis, the same numerical instability prevales when the element Peclet number is large. Thus a SUPG formulation for the eigenvalue problem is proposed to eliminate the instability.
  
  Similiar to the classical SUPG method, we take the inner-product of both sides of equation~(\ref{E:eigenNS1}) with $\mathbf{w}_h$ and with $\tau \mathbf{u}_0\cdot\nabla\mathbf{w}_h$ in element internal domain $\tilde{\Omega}=\cup\Omega^e$, where $\Omega^e$ is the internal of the eth element. It follows that the discretized weak form is,
   \begin{align}
   		\nonumber find~&\mathbf{u}_h \in\mathbf{V}_h^0,~\lambda\in\mathbb{C}~s.t, \label{E:weakformSUPG}\\
   		\lambda M(\mathbf{u}_h,\mathbf{w}_h) =& A(\mathbf{u}_h,\mathbf{w}_h)+C(\mathbf{u}_h,\mathbf{w}_h),~\forall~\mathbf{w}_h\in \mathbf{V}_h^0, \\
   		\nonumber M(\mathbf{u}_h,\mathbf{w}_h) =& \left<u_h^i,w_h^i\right>+\left<u^i,\tau u_0^k\nabla_kw_h^i \right>_{\tilde{\Omega}}, \\
   		\nonumber A(\mathbf{u}_h,\mathbf{w}_h) =& -\nu \left<\nabla_k u_h^i, \nabla_k w_h^i \right>+\nu\left\langle \nabla^2u_h^i, \tau u_0^k\nabla_kw_h^i \right\rangle_{\tilde{\Omega}} , \\
   		\nonumber C(\mathbf{u}_h,\mathbf{w}_h) =& \left< -u_0^k \nabla_k u_h^i - u_h^k \nabla_k u_0^i, w_h^i\right> \\ \nonumber
   		&+\left<-u_0^k \nabla_k u_h^i - u_h^k\nabla_k u_0^i,\tau u_0^k\nabla_kw_h^i \right>_{\tilde{\Omega}}, \\
   		\nonumber ||\mathbf{u_h}||_{L^2(\Omega)} =& 1. 
   \end{align}
    The parameter $\tau$ can be chosen as $\tau=min(\alpha/3,1)h/(2|\mathbf{u}_0|)$, see~\cite{hughesbeyondsupg}. 
    
    Appendix gives the convergence analysis of this new scheme. Convergence rate is shown to be $O(h^2)$ for eigenvalue and $O(h)$ for eigenmode energy norm.
    
    The weak form~(\ref{E:weakformSUPG}) is equivalent to a generalized eigenvalue problem,
    \begin{align}
    	H\mathbf c&=\lambda M\mathbf c, \label{E:eigenmatrix} \\
    	\nonumber  H_{ij} &= A(\mathbf\Phi_j, \mathbf\Phi_i)+C(\mathbf\Phi_j, \mathbf\Phi_i), \\
    	\nonumber M_{ij} &= M(\mathbf\Phi_j, \mathbf\Phi_i).
    \end{align}
    The eigenmode $\mathbf{u}$ to be solved is $\mathbf{u}=\sum_{i}c_i\mathbf\Phi_i$.
  
  The implicitly-restarted Arnoldi method, provided by ARPACK, or MATLB eigs function is capable of finding the least-stable eigenvalues and eigenmodes.

\section{Numerical Results}\label{S:NumResult}
We now show some of the computed results of the new numerical scheme.

$Re$ from $500-2500$ are computed for a finite-length duct with nondimensional duct length $L=3$. Regular mesh size of $N_x=N_y=Nz/L$ is adopted. 

Figure~\ref{F:meshconvergence} shows the mesh convergence plot. Computed least-stable eigenvlaue $\lambda_1,\lambda_2$ and the third eigenvalue $\lambda_3$ tend to converge as the degrees of freedom increases. The following analyasis will all adopt the case with $N_x=40$, which is the largest degree of freedom computed so far.
\begin{figure}
\centering
\includegraphics[width=2.22in]{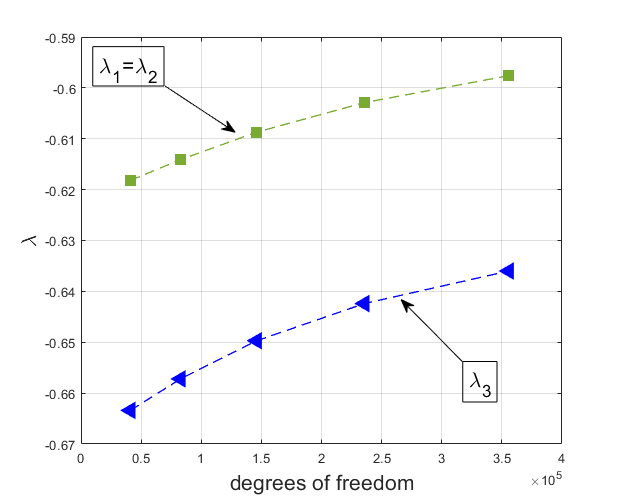}
\caption{Convergence plot of leading eigenvalues}
\label{F:meshconvergence}
\end{figure}

The computed eigenvalues show that the largest eigenvalue has multiplicity of $2$. Acutually the first two eigenmodes are identical by interchanging $x,y$ coordinates and $u,v$ components. The third eigenmode has a multiplicity of $1$ and this eigenmode possesses $x$-$y$ symmetry.

Figure~\ref{F:lambdavsRe} shows the three least-stable eigenvalues as a function of $Re$.
It is again found that the duct Poiseuille flow is linearly stable for all $Re$ up to $2500$, which agrees with previous pipe or duct theoretical studies.
\begin{figure}
	\centering
	\includegraphics[width=2.22in]{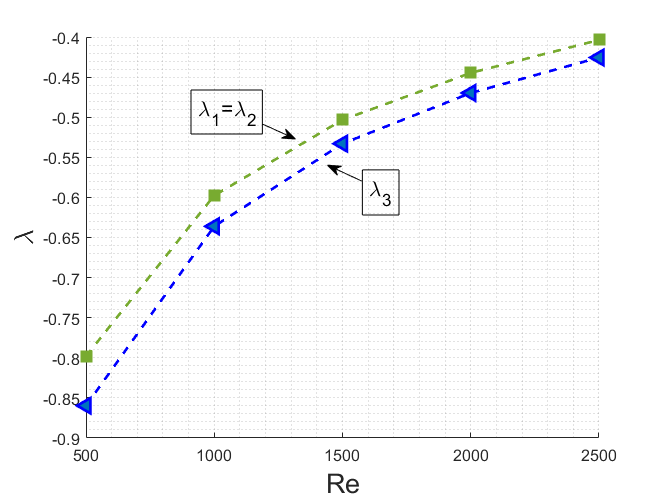}
	\caption{Least-stable eigenvalues as a function of $Re$}
	\label{F:lambdavsRe}
\end{figure}

It is interesting to investigate the leading-order eigenmode structures.
Figure~\ref{F:modes1} shows the normalized eigenmode corresponding to the first eigenvalue $\lambda_1=-0.4448$ at $Re=2000$. Figure~\ref{F:modes3} shows the normalized eigenmode corresponding to the third eigenvalue $\lambda_3=-0.4698$ at the same $Re$. Mesh size with $N_x\times N_y \times N_z = 40\times 40 \times 120$ are used to compute the solutions. The cut planes are chosen at places where the mode is dominant. The second mode is not plotted because it is a symmetric mode of the first one. Although complicated in the structures of the shape, there is a characteristic boundary-layer-structure that prevails at all high $Re$ leading-order modes shapes. The modes are dominant in the vicinity of the wall and are convected downstream. When the pipe inlet profile is fixed to the base flow profile $\mathbf{u}_0$, least-stable perturbations tends to dominate near the outlet. It can also be observed that these modes do not have a simple normal mode structure, e.g.,
\begin{align}
	\mathbf{u}(x,y,z)=\mathbf{\Psi}(x,y)e^{ikz},
\end{align}    
where $k$ is the wave number and $\Psi$ is a common shape function. Therefore this simple theoretical assumption may not be the best way to analyze pipe or duct flow stability.
\begin{figure}
	\centering
	\includegraphics[width=5.5in]{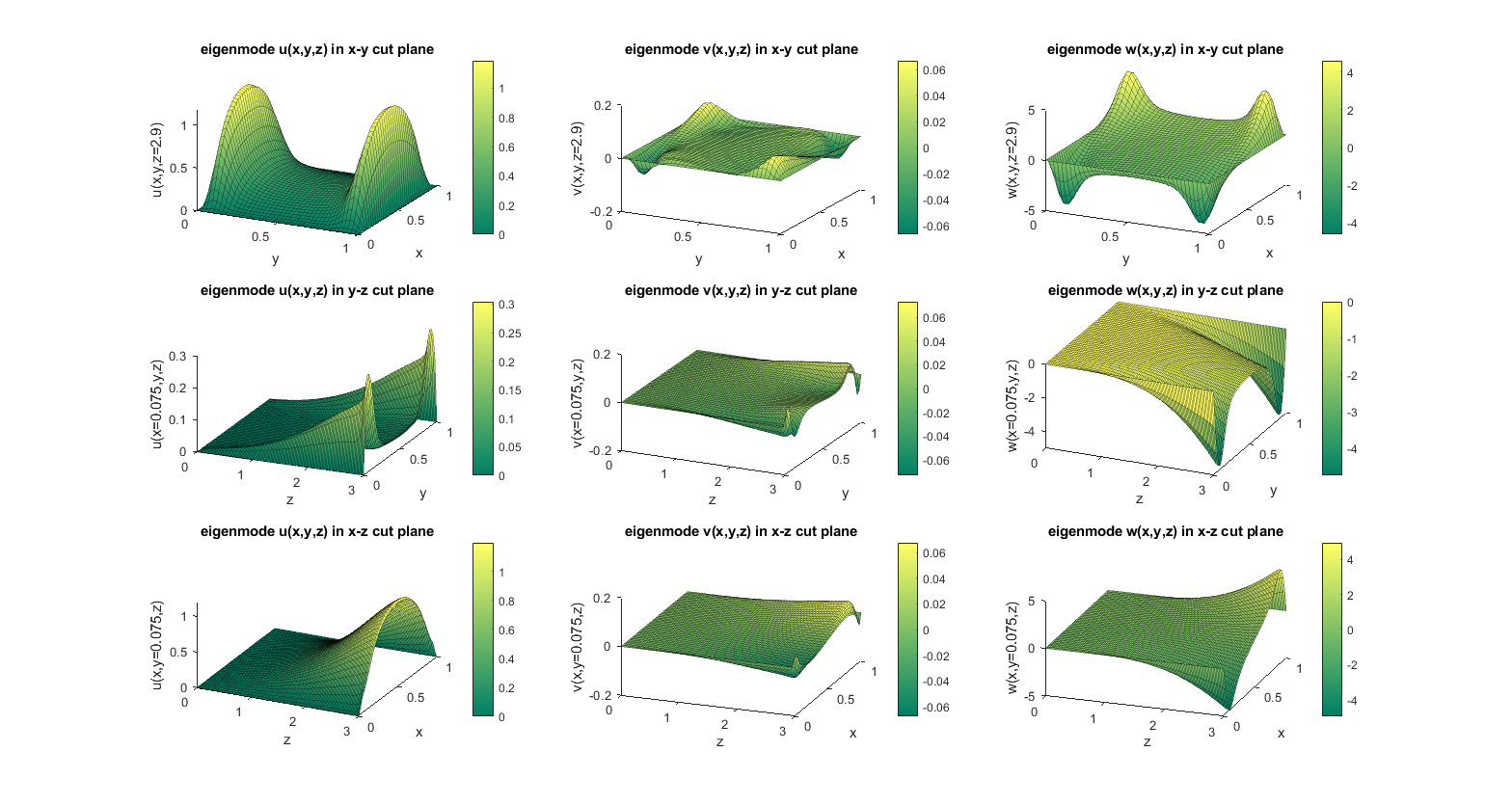}
	\caption{Eigenmode corresponding to the first eigenvalue $\lambda_1=-0.4448$ at $Re=2000$, nondimensional duct length is $L=3$ with mesh $N_x\times N_y\times N_z=40\times 40\times 120$. Computed mode is normalized, $||\mathbf{u}_1||=1$. First row, eigenmode $u,v,w$ in $x$-$y$ cut plane at $z=2.9$. Second row, eigenmode $u,v,w$ in $y$-$z$ cut plane at $x=0.075$. Third row, eigenmode $u,v,w$ in $x$-$z$ cut plane at $y=0.075$.}
	\label{F:modes1}
\end{figure}

\begin{figure}
	\centering
	\includegraphics[width=5.5in]{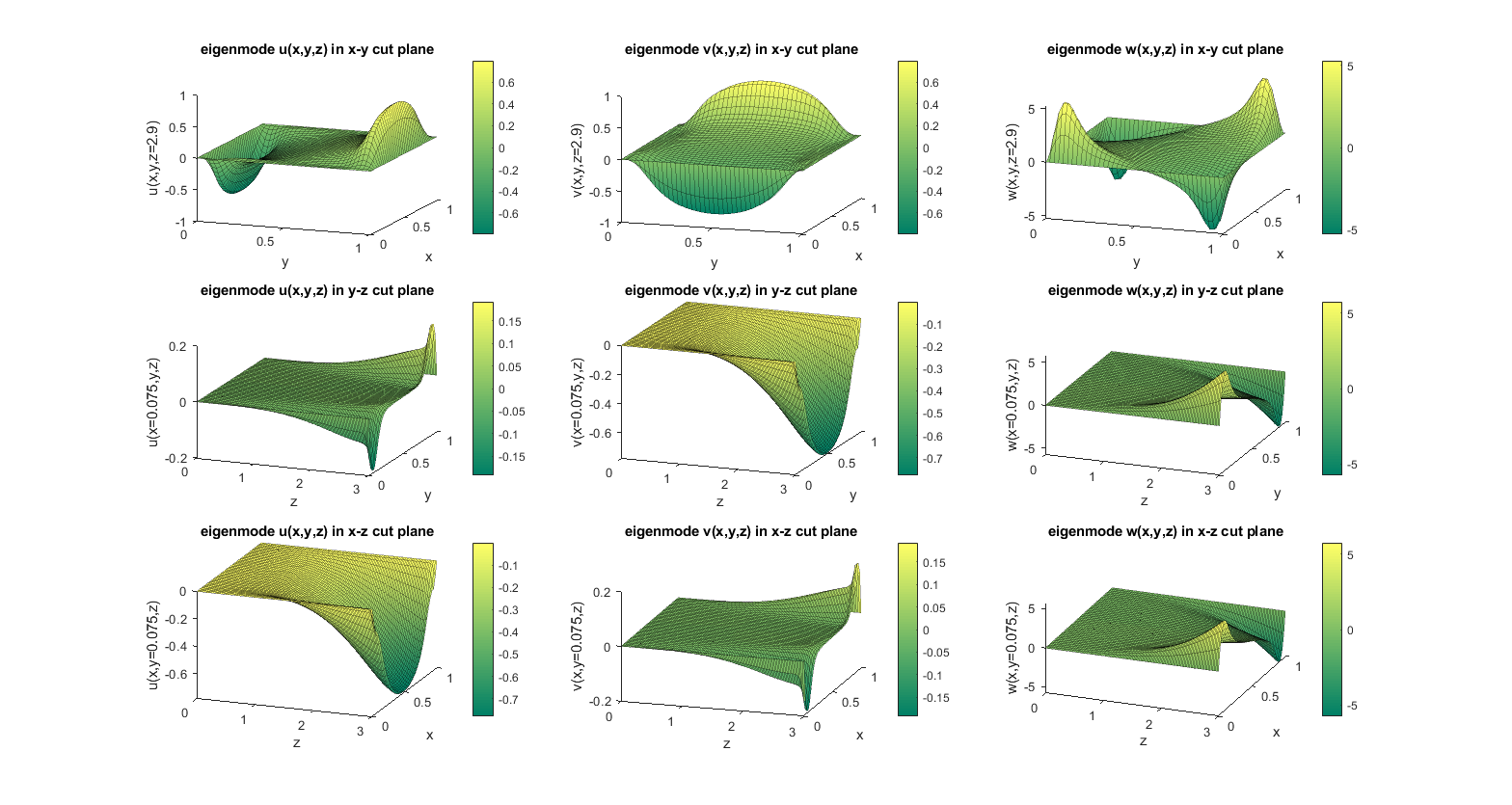}
	\caption{Eigenmode corresponding to the third eigenvalue $\lambda_3=-0.4698$ at $Re=2000$, nondimensional duct length is $L=3$ with mesh $N_x\times N_y\times N_z=40\times 40\times 120$. Computed mode is normalized, $||\mathbf{u}_3||=1$. First row, eigenmode $u,v,w$ in $x$-$y$ cut plane at $z=2.9$. Second row, eigenmode $u,v,w$ in $y$-$z$ cut plane at $x=0.075$. Third row, eigenmode $u,v,w$ in $x$-$z$ cut plane at $y=0.075$}
	\label{F:modes3}
\end{figure}

\begin{figure}
	\centering
	\includegraphics[width=2.5in]{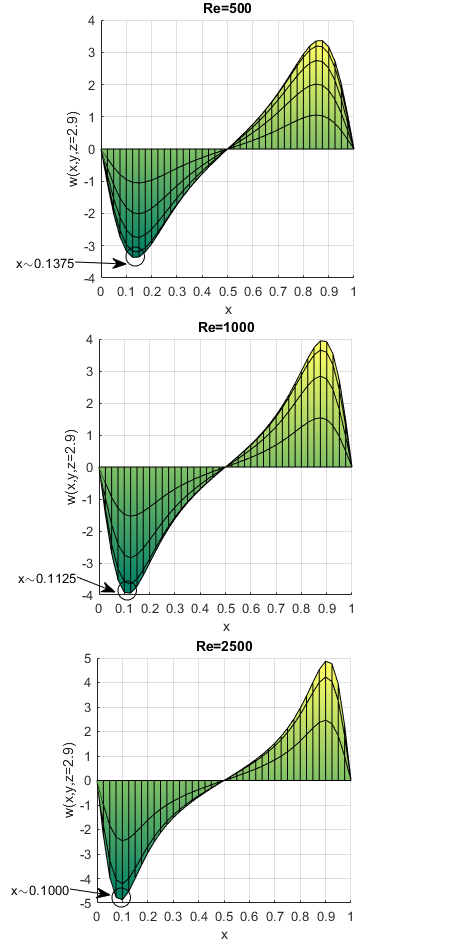}
	\caption{Side view of the first eigenmode $w(x,y,z=2.9)$. Top figure, mode at $Re=500$. Central figure, mode at $Re=1500$. Bottom figure, mode at $Re=2500$. Position $x$ where the mode peaks are labelled.}
	\label{F:blstruc}
\end{figure}
The boundary layer gets thinner and closer to the duct wall as $Re$ increases. The gradient of the normalized modes become unbounded as $Re$ grows. Figure~\ref{F:blstruc} shows the side view of the first eigenmode $w(x,y,z=2.9)$ corresponding to $\lambda_1$ at $Re=500,1000,2500$. The position where the mode function peaks are labelled. It can be seen the peak position from the duct wall decreases from $x\sim0.1375$ at $Re=500$ to $x\sim0.1$ at $Re=2500$. The physical consequences are discussed in the next section. It seems true that 'instability' tends to be generated near the wall and the outlet in the very beginning phase of transition.

\section{Relationship with pseudospectra theory and nonlinear Effects}\label{S:NonlinearEffects}
The relationship with pseduospectra theory and nonlinear effects of Poiseuille flow is studied. 

It has been investigated in~\cite{trefethenscience} that the instability of pipe Poiseuille is  attributed to the sensitivity of the linear operator $\mathcal{L}$ from the linearization of the Navier-Stokes equations, where $\mathcal{L}\mathbf{u}=-\mathbf{u_0} \cdot \nabla \mathbf{u} -\mathbf{u} \cdot \nabla \mathbf{u_0}  -\nabla p + \nu \nabla^2 \mathbf{u}$. It was shown that when $Re$ increase, a small perturbation to the linear operator $\mathcal{L}$ results in a relatively large eigenvalue deviation. As a result, pipe Poiseuille flow may become unstable at finite $Re$ although in ideal case it is linearly stable. We connect this claim to the boundary-layer-structure of the least-stable eigenmode, in a qualitative way. The boundary-layer-structure has most of its active part near the duct wall. As $Re$ becomes higher, the layer becomes thinner and closer to the wall. It is expected that a small fluctuation of the duct wall configuration will alter the eigenmode shape in a significant way. The base flow field, however, is not affected greatly by wall fluctuations since its velocity field is dominant at the bulk. With this physical picture in mind, a perturbation to $\mathcal{L}$ induced by small fluctuation of duct wall will change the eigenvalue noticeablely as a result of the singularity of mode shapes in the limit of high $Re$. This gives a physical point of view regarding the pseudospectra theory.

To investigate the nonlinear effects, we assume the initial condition to be one of the least-stable modes $\epsilon\mathbf{u}_1$, where $||\mathbf{u}_1||=1$ corresponding to $\lambda_1$. The instananeous flow energy growth is investigated. Physically it can be regarded as the study of the nonlinear dynamics for a short period of time. Multiplying $\mathbf{u}$ into equation~(\ref{E:unsteadyNS}) and integrate on both sides gives,
    \begin{align}
    \frac{1}{2} \frac{\partial}{\partial t}\left<\mathbf{u},\mathbf{u}\right> = A_0(\mathbf{u},\mathbf{u})+C_0(\mathbf{u},\mathbf{u}) - \left<\mathbf{u}\cdot\nabla\mathbf{u},\mathbf{u}\right>,
    \end{align}
    where $A_0$, $C_0$ are defined in equation~(\ref{E:weakform1}).
    Since the initial condition is chosen to be the least-stable eigenmode $\epsilon\mathbf{u}_1$, it follows,
    \begin{align}
    	\frac{1}{2} \frac{\partial}{\partial t}\left<\mathbf{u},\mathbf{u}\right>\bigg|_{t=0} = 
    	\epsilon^2 \lambda_1\left<\mathbf{u}_1,\mathbf{u}_1\right> - \epsilon^3 \left<\mathbf{u}_1\cdot\nabla\mathbf{u}_1,\mathbf{u}_1\right>,
    \end{align}
    where $ \lambda_1$ is the least-stable eigenvalue. Thus the total flow energy $\frac{1}{2}\left<\mathbf{u},\mathbf{u}\right>$ growth rate at instant $t=0$ is driven by the stabilizing linear effect and the nonlinear effect.  By switching the sign of the perturbation amplitude $\epsilon$, the nonlinear energy production rate can always be positive,  we therefore consider the nonlinear term to be destabilizing in the worst case.
    
    If the nonlinear effect is negligible, the flow energy decays.  The assumption to take the eigenmode as the initial condition physically makes sense when flow is attracted to the Poiseuille base flow branch. On the other hand, when $\epsilon$ is large, the nonlinear energy production rate may exceed the linear stabilizing effect and the flow energy starts to grow at $t=0$. The linearization assumption breaks down in this situation. Thus we study the critical amplitude $\epsilon$ at various $Re$. Perturbation with above the critical value may not result in a final transistion into turbulence since nonlinear dynamics after $t=0$ will be complicated and is not analyzed. However, it is a sign that linear stability problem fails to give a proper mathematical model and the flow tends to be dominated by nonlinear effects. With the  difficulty in quantifying subsequent nonlinear dynamics, the analysis only gives a rough estimate on how small perturbations may generate  non-negligible destablizing effects. It should be noted that a eigenmode initial perturbation is different from the other studies to induce impulsive or periodic distrubances and this is not very likely to be physically induced.
    
    The critical amplitude is determined by a differential balance equation,
    \begin{align}
    	  \epsilon = \inf \frac{|\lambda_1 u_1^i u_1^i d\Omega|} {|u_1^k\nabla_ku_1^iu_1^i d\Omega|}.
    \end{align}
    When the energy production term $|u_1^k\nabla_ku_1^iu_1^i d\Omega|$ exceeds the linear stablizing term $\lambda_1 u_1^i u_1^i d\Omega$ in any infinitesimal volume, the linearzation assumption is no longer valid. Thus the critical amplitude of the flow energy can be computed numerically by finite element integration on each hexahedral element. Figure~\ref{F:nonlinear} shows the critical perturbation amplitude $\epsilon$ vs $Re$, the effective nonlinear energy production $\eta=\sup\frac{|u_1^k\nabla_ku_1^iu_1^i d\Omega|}{|\lambda_1 u_1^i u_1^i d\Omega|}$ vs $Re$, and the first eigenvalue $\lambda_1$ vs $Re$. The threshold is seen to be decreasing greatly as $Re$ increases. Any reasonable perturbation at very high $Re$ tend to have non-negligible nonlinear effect. A fit to the $\epsilon-Re$ curve gives approximately $\epsilon\sim O(Re^{-0.6})$. This number is different from the results reported by~\cite{hof},~\cite{meseguernonlin} and~\cite{eckhardt}, mainly because the initial condition given here is not applicable to the other studies. It may also attribute to the difference between pipe and duct flow.  The significance of the nonlinear analysis demonstrates that the linear stability assumption fails to be a good model problem as $Re$ becomes higher. Two reasons contribute to the breakdown. First, the decreasing absolute value of first eigenvalue as $Re$ increass. Second, the increasing nonlinear energy production rate as a result of the growing eigenfunction gradient near the duct wall, as it is found that the maximum nonlinear energy production is near the peak of the boundary layer.
    
    \begin{figure}
    	\centering
    	\includegraphics[width=5.5in]{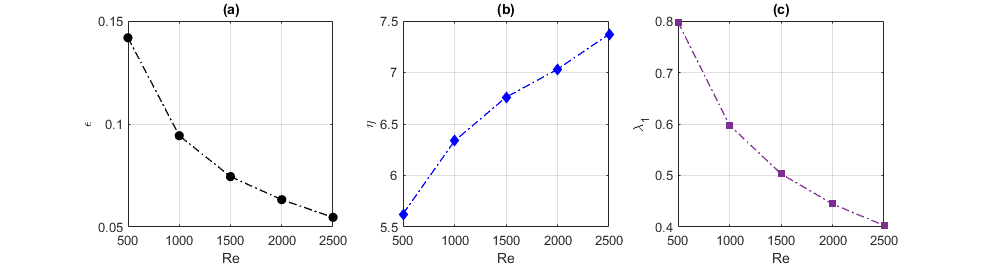}
    	\caption{(a). Critical perturbation amplitude $\epsilon$ vs. $Re$. (b). Effective nonliner energy production rate $\eta$ vs. $Re$. (c). Leading-order eigenvalue $\lambda_1$ vs. $Re$. }
    	\label{F:nonlinear}
    \end{figure}

\section{Conclusions}\label{S:conclusion}
In this paper have presented a SUPG based finite element method with divergence-free-basis technique to compute the Poiseuille flow eigenvalue problem. The flow is found to be linearly stable for $Re$ up to $2500$ and the modes show a boundary-layer-structure at high $Re$. This characteristic structure  at high $Re$ accounts for the sensitivity of the linearized Navier-Stokes problem with respect to geometry variation as well as the sensitivity of nonlinear effects.

\section{Appendix}
The convergence analysis in this paper assumes that the domain boundary $\partial{\Omega}$ is smooth enough, and also that the base flow $\mathbf{u}_0$, the eigenmode $\mathbf{u}$ and the corrresponding pressure field $p$ are smooth enough, for simplicity. 

Space $\mathbf{V}=\left\lbrace  \mathbf{u}\in [H^1(\Omega)]^3 |~ \mathbf{u} ~satisfies ~boundary~condition~(\ref{E:eigenBCinlet})-(\ref{E:eigenBCoutlet}) \right\rbrace $ and  $Q=\left\lbrace p| p\in L^2(\Omega)/\mathbb{R} \right\rbrace $ are now assumed to be complex spaces. $\mathbf{V}_h \subset \mathbf{V}$, $Q_h \subset Q$, $\mathbf{V}^0 = \left\lbrace  \mathbf{u} \in \mathbf{V} | \nabla \cdot \mathbf{u}=0 \right\rbrace$ and $\mathbf{V}_h^0 = \left\lbrace \mathbf{u}_h \in \mathbf{V}_h |	\left\langle \nabla \cdot \mathbf{u}_h, q_h \right\rangle =0, ~\forall q_h \in Q_h \right\rbrace$ also extend to complex spaces accordingly.

\begin{lemma} \label{L:coerce1}
	Let 
	\begin{align}
		\nonumber	a(\mathbf{u},\mathbf{w}) &= -A_0(\mathbf{u},\mathbf{w})+\gamma\left<\mathbf{u},\mathbf{w}\right>\\
		\nonumber	c(\mathbf{u},\mathbf{w}) &= -C_0(\mathbf{u},\mathbf{w}),
	\end{align}
	where $\mathbf{u},\mathbf{w} \in \mathbf{V}$. $A_0$, $C_0$ are defined in equation~(\ref{E:weakform1}).  There exists a sufficiently large $\gamma > 0$ such that,
	\begin{align}
		\nonumber \exists~\alpha_0>0,~a(\mathbf{u},\mathbf{u})+c(\mathbf{u},\mathbf{u}) \geq \alpha_0 ||\mathbf{u}||_1^2
	\end{align}
\end{lemma}
\begin{proof}
	First we have,
	\begin{equation}
	-A_0(u,u) = \nu\left< \nabla_k u^i, \nabla_k u^i \right> > 0. \label{E:viscint1}
	\end{equation}
	
	Next $\left< \mathbf{u}_0 \cdot \nabla \mathbf{u}, \mathbf{u} \right>$ is estimated,
	
	\begin{align}
		\left< \mathbf{u}_0 \cdot \nabla \mathbf{u}, \mathbf{u} \right> &= 
		\int_{\Omega} u_0^k\nabla_ku^iu^i d\Omega  \label{E:convecint1} \\
		\nonumber &=	\int_{\Omega}\nabla_k( u_0^ku^iu^i) d\Omega - 	\int_{\Omega}  u_0^ku^i\nabla_ku^i d\Omega.
	\end{align}
	Here $\nabla_ku_0^k=0$ is implied.
	From equation~(\ref{E:convecint1}),
	\begin{align}
		\left< \mathbf{u}_0 \cdot \nabla \mathbf{u}, \mathbf{u} \right> &= 
		\frac{1}{2}	\int_{\Omega}\nabla_k( u_0^ku^iu^i) d\Omega  \label{E:convecint2} \\
		\nonumber &= \frac{1}{2} \int_{\partial \Omega}  u_0^ku^iu^in^k d\sigma.
	\end{align}
	
	Here $\mathbf{n}$ is the unit outward normal vector. Since the problem assumes outflow at the outlet, namely $u_0^kn^k|_{outlet}>0$, and at the inlet and wall $u^i=0$, it follows that,
	\begin{align}
		\left< \mathbf{u}_0 \cdot \nabla \mathbf{u}, \mathbf{u} \right> > 0.
	\end{align}
	
	Then $\left< \mathbf{u} \cdot \nabla \mathbf{u}_0, \mathbf{u} \right>$ is estimated.
	Since the base flow field $\mathbf{u}_0$ and its gradient is bounded, there exists a constant tensor $N_{ki}$ such that,
	\begin{align}
		|\nabla_ku_0^i| \leq N_{ki}.
	\end{align}
	Then,
	\begin{align}
		\left| \left< \mathbf{u} \cdot \nabla \mathbf{u}_0, \mathbf{u} \right>\right|  &= 
		\left| \int_{\Omega} u^k\nabla_ku_0^iu^i d\Omega\right|   \label{E:convecint3} \\
		\nonumber & \leq \int_{\Omega}N_{ki}|u^k||u^i|d\Omega \\
		\nonumber & < \gamma ||u||^2,
	\end{align}
	Here $\gamma$ is a sufficiently large positive number. It is apparant that when $\gamma>0$ is large enough, matrix $\gamma I-N $ is positive definite.
	
	It then follows,
	\begin{align}
		\exists~\alpha_1 > 0,~\left< \mathbf{u} \cdot \nabla \mathbf{u}_0, \mathbf{u} \right> + \gamma\left<\mathbf{u},\mathbf{u}\right> > \alpha_1 \left<\mathbf{u},\mathbf{u}\right>. \label{E:convecint4}
	\end{align}
	
	Combining equations~(\ref{E:viscint1}),~(\ref{E:convecint1}),~(\ref{E:convecint4}), $\exists~\alpha_0>0$, 
	$a(\mathbf{u},\mathbf{u})+c(\mathbf{u},\mathbf{u}) \geq \alpha_0 ||\mathbf{u}||_1^2$.
	
\end{proof}

\begin{lemma}\label{L:coercivitySUPG}
	Let $\Omega^e$ be the interior of the eth element, $\tilde{\Omega}=\cup\Omega^e$, and,
	\begin{align}
		&a_{supg}(\mathbf{u}_h,\mathbf{w}_h)+c_{supg}(\mathbf{u}_h,\mathbf{w}_h)  \\
		\nonumber =&\left<\nu \nabla\mathbf{u}_h,\nabla\mathbf{w}_h\right> +  \left< \mathbf{u}_0\cdot\nabla\mathbf{u}_h,\mathbf{w}_h\right> \\ 
		\nonumber &+\left< \mathbf{u}_h\cdot\nabla\mathbf{u}_{0},\mathbf{w}_h\right>+\left< \gamma\mathbf{u}_h,\mathbf{w}_h\right>\\
		\nonumber	&-\left< \nu\nabla^2\mathbf{u}_h,\tau \mathbf{u}_0\cdot\nabla\mathbf{w}_h\right>_{\tilde{\Omega}} + \left< \mathbf{u}_0\cdot\nabla\mathbf{u}_h,\tau \mathbf{u}_0\cdot\nabla\mathbf{w}_h\right>_{\tilde{\Omega}} \\
		\nonumber &+
		\left< \mathbf{u}_h\cdot\nabla\mathbf{u}_{0},\tau \mathbf{u}_0\cdot\nabla\mathbf{w}_h\right>_{\tilde{\Omega}} +
		\left< \gamma\mathbf{u}_h,\tau \mathbf{u}_0\cdot\nabla\mathbf{w}_h\right>_{\tilde{\Omega}}.
	\end{align}
	where $\mathbf{u}_h,\mathbf{w}_h \in \mathbf{V}_h$.  $\mathbf{V}_h$ is the discretized continuous piecewise polynomial solution space that satifies boundary conditions~(\ref{E:eigenBCinlet})-(\ref{E:eigenBCoutlet}).
	\begin{align}
		\nonumber	\tau&=\frac{h}{2|\mathbf{u}_0|}\zeta(\alpha),\\
		\nonumber \alpha&= \frac{h|\mathbf{u}_0|}{2\nu}.
	\end{align}
	Here $\alpha$ is the element Peclet number. $\zeta(\alpha)$ satisfies
	\begin{align}
		\nonumber	m=\sup_{\alpha}\frac{\zeta(\alpha)}{\alpha} \leq \frac{4}{C^2},
	\end{align}
	where $C$ is the inverse estimate constant that satifies,
	\begin{align}
		||\Delta\mathbf{v}_h||_{\tilde{\Omega}} \leq Ch^{-1}||\nabla \mathbf{v}_h||_{\Omega}. \label{E:inverseestimate}
	\end{align}     
	$h$ is the mesh parameter. In the regular cubic element case, it is the edge length. $|\mathbf{u}_0|$ denotes the maximum of base flow velocity magnitude within each element. $\tau$ varies across each element. The subscript $\tilde{\Omega}$ means that the inner-product is taken in the interior of each element and then summing the elementwise integration up. 
	
	Then $\exists~\gamma>0$ defined in lemma~\ref{L:coerce1}, and $\exists~\delta>0$, when $\tau < \delta$, the following coercive condition holds,
	\begin{align}
		\nonumber \exists~\alpha_0>0,~a_{supg}(\mathbf{u}_h,\mathbf{u}_h)+c_{supg}(\mathbf{u}_h,\mathbf{u}_h) \geq \alpha_0 ||\mathbf{u}_h||_1^2+ \frac{1}{6}||\tau^{1/2} \mathbf{u}_0\cdot\nabla\mathbf{u}_h||_{\tilde{\Omega}}^2.
	\end{align}
\end{lemma}

\begin{proof}
	\begin{align}
		&a_{supg}(\mathbf{u}_h,\mathbf{u}_h)+c_{supg}(\mathbf{u}_h,\mathbf{u}_h) \label{E:supgcoerce2} \\
		\nonumber =&\left<\nu \nabla\mathbf{u}_h,\nabla\mathbf{u}_h\right> +  \left< \mathbf{u}_0\cdot\nabla\mathbf{u}_h,\mathbf{u}_h\right> \\ 
		\nonumber &+\left< \mathbf{u}_h\cdot\nabla\mathbf{u}_{0},\mathbf{u}_h\right>+\left< \gamma\mathbf{u}_h,\mathbf{u}_h\right>\\
		\nonumber	&-\left< \nu\nabla^2\mathbf{u}_h,\tau \mathbf{u}_0\cdot\nabla\mathbf{u}_h\right>_{\tilde{\Omega}} + \left< \mathbf{u}_0\cdot\nabla\mathbf{u}_h,\tau \mathbf{u}_0\cdot\nabla\mathbf{u}_h\right>_{\tilde{\Omega}} \\
		\nonumber &+
		\left< \mathbf{u}_h\cdot\nabla\mathbf{u}_{0},\tau \mathbf{u}_0\cdot\nabla\mathbf{u}_h\right>_{\tilde{\Omega}} +
		\left< \gamma\mathbf{u}_h,\tau \mathbf{u}_0\cdot\nabla\mathbf{u}_h\right>_{\tilde{\Omega}}.
	\end{align}
	
	First we recall the inverse estimate relationship in equation~(\ref{E:inverseestimate}), 
	\begin{align}
		&\nu\tau = \frac{\nu h}{2|\mathbf{u}_0|}\zeta(\alpha)=\frac{1}{4}h^2\frac{\zeta(\alpha)}{\alpha} \leq \frac{1}{4}h^2m \leq \frac{h^2}{C^2} \\
		\nonumber \Rightarrow &~ \nu^2||\tau^{1/2}\nabla^2 \mathbf{u}_h||^2_{\tilde{\Omega}} \leq \nu||\nabla\mathbf{u}_h||^2
	\end{align}
	
	With this inequality, $-\left< \nu\nabla^2\mathbf{u}_h,\tau\mathbf{u}_0\cdot\nabla\mathbf{u}_h\right>_{\tilde{\Omega}}$ is estimated,
	\begin{align}
		&|-\left< \nu\nabla^2\mathbf{u}_h,\tau \mathbf{u}_0\cdot\nabla\mathbf{u}_h\right>_{\tilde{\Omega}}| \\
		\nonumber \leq& \frac{3}{4} \nu^2||\tau^{1/2}\nabla^2 \mathbf{u}_h||^2_{\tilde{\Omega}} + \frac{1}{3} ||\tau^{1/2}\mathbf{u}_0\cdot\nabla\mathbf{u}_h||^2_{\tilde{\Omega}} \\
		\nonumber \leq& \frac{3}{4} \nu||\nabla \mathbf{u}_h||^2 + \frac{1}{3} ||\tau^{1/2}\mathbf{u}_0\cdot\nabla\mathbf{u}_h||^2_{\tilde{\Omega}}.
	\end{align}
	
	Then $\left<\mathbf{u}_h\cdot\nabla\mathbf{u}_{0},\tau \mathbf{u}_0\cdot\nabla\mathbf{u}_h\right>_{\tilde{\Omega}} $ is estimated,
	\begin{align}
		&| \left<\mathbf{u}_h\cdot\nabla\mathbf{u}_{0},\tau \mathbf{u}_0\cdot\nabla\mathbf{u}_h\right>_{\tilde{\Omega}} | \\
		\nonumber \leq & ||\tau^{1/2}\mathbf{u}_h\cdot\nabla\mathbf{u}_0||^2 + \frac{1}{4}||\tau^{1/2}\mathbf{u}_0\cdot\nabla\mathbf{u}_h||^2_{\tilde{\Omega}} \\
		\nonumber \leq & ||N\tau^{1/2}\mathbf{u}_h||^2+\frac{1}{4}||\tau^{1/2}\mathbf{u}_0\cdot\nabla\mathbf{u}_h||^2_{\tilde{\Omega}}.
	\end{align}
	Here $N$ is a constant due to the boundedness of $\nabla \mathbf{u}_0$.
	
	Next we estimate $ \left< \gamma\mathbf{u},\tau \mathbf{u}_0\cdot\nabla\mathbf{u}_h\right>_{\tilde{\Omega}}$,
	\begin{align}
		&|\left<\gamma\mathbf{u}_h,\tau \mathbf{u}_0\cdot\nabla\mathbf{u}_h\right>_{\tilde{\Omega}}| \\
		\nonumber \leq & ||\gamma\tau^{1/2}\mathbf{u}_h||^2 + \frac{1}{4}||\tau^{1/2}\mathbf{u}_0\cdot\nabla\mathbf{u}_h||^2_{\tilde{\Omega}}.
	\end{align}

	Subsituting all the above estimations into equation~(\ref{E:supgcoerce2}) and use the results in lemma~\ref{L:coerce1},
	\begin{align}
		&a_{supg}(\mathbf{u}_h,\mathbf{u}_h)+c_{supg}(\mathbf{u}_h,\mathbf{u}_h) \\
		\nonumber \geq& \frac{\nu}{4}||\nabla \mathbf{u}_h||^2 + \frac{1}{6}||\tau^{1/2}\mathbf{u}_0\cdot\nabla\mathbf{u}_h||_{\tilde{\Omega}}^2+ \left< (\alpha_1-N^2\tau-\gamma^2\tau)\mathbf{u}_h,\mathbf{u}_h\right>_{\tilde{\Omega}}.
	\end{align}
	Here $\alpha_1$ is the estimation constant in lemma~\ref{L:coerce1}.
	It follows that when
	\begin{align}
		\tau <&   \frac{\alpha_1}{N^2+\gamma^2}=\delta, \\
		\nonumber \exists~\alpha_0>0,&~a_{supg}(\mathbf{u}_h,\mathbf{u}_h)+c_{supg}(\mathbf{u}_h,\mathbf{u}_h) \geq \alpha_0 ||\mathbf{u}_h||_1^2 + \frac{1}{6}||\tau^{1/2} \mathbf{u}_0\cdot\nabla\mathbf{u}_h||_{\tilde{\Omega}}^2.
	\end{align}

\end{proof}

\begin{remark}
	For typical flow problem $\alpha_1, \gamma$ can be chosen on the order of $N\sim|\nabla\mathbf{u}_0|$, so the condition for $\tau$ is not restrictive.
\end{remark}

\begin{remark}
	The introduction of '$\gamma$' terms ensures the coercivity of the linear problem. The corresponding eigenvalue problem is,
	\begin{align}
		a_{supg}(\mathbf{u},\mathbf{w})+c_{supg}(\mathbf{u},\mathbf{w})&=m_{supg}(\tilde{\lambda}\mathbf{u},\mathbf{w}), \label{E:lambdaproblem} \\
		\nonumber m_{supg}(\mathbf{u},\mathbf{w}) &= -M(\mathbf{u},\mathbf{w}).
	\end{align}
	Here $M(\mathbf{u},\mathbf{w})$ is a bilinear form defined in equation~%(\ref{E:weakformSUPG).
	
	The discretized problem is equivalent to a matrix eigenvalue problem,
	\begin{align}
		H\mathbf c- \gamma M\mathbf c = \tilde{\lambda}_hM\mathbf c.
	\end{align}
	Here $H$, $M$ are matrices defined in equation~(\ref{E:weakformSUPG}).
	It is apparant that the eigenvalue is shifted by $\gamma$ from the original problem~\ref{E:eigenmatrix}, and that the eigenfunctions remain unchanged.  
\end{remark}

\begin{lemma}
	Let
	\begin{align}
		a_0(\mathbf{u},\mathbf{w}) &= \left<\nabla_ku^i,\nabla_kw^i\right>,~\forall~\mathbf{u},\mathbf{w}\in \mathbf{V}, \\ \nonumber
		b(\mathbf w,p) &= \left< \nabla\cdot\mathbf w, p\right>,~\forall~\mathbf{w}\in\mathbf{V},p\in Q.
	\end{align}
	Define a linear system in its weak form, $\mathbf{f}\in \mathbf{V}$,
	\begin{align}
		a_0(\mathbf{u},\mathbf{w})+b(\mathbf w,p) &= \left<\mathbf{f},\mathbf w\right>,~\forall\mathbf{w} \in \mathbf{V},\label{E:LBB-1}\\
		\nonumber
		b(\mathbf u,q) &= 0,~\forall q\in Q,
	\end{align}
	Here the binlinear form $b(\mathbf{u},q)$ satisfies the B.B. condition such that,
	\begin{align}
		\nonumber &\exists~ \beta > 0,~ 
		\inf_{q\in Q, q\ne 0} \sup_{\mathbf{u}\in\mathbf{V},\mathbf{u}\ne\mathbf{0}} \frac{|b(\mathbf{u},q)|}{||\mathbf{u}||_{\mathbf{V}} ||q||_Q} \geq \beta.
	\end{align}
	The corresponding discretized weak form is,
	\begin{align}
		a_0(\mathbf{u}_h,\mathbf{w}_h)+b(\mathbf w_h,p_h) &= \left<\mathbf{f}_h,\mathbf w_h\right>,~\forall\mathbf{w}_h \in \mathbf{V}_h,\label{E:LBB0}\\
		\nonumber
		b(\mathbf u_h,q_h) &= 0,~\forall q_h\in Q_h.
	\end{align}
	Then whether or not the B.B. condition is satisfied or not in equations~(\ref{E:LBB0}), there is a unique $\mathbf{u}_h$ that satifies the discretized equations and,
	\begin{align}
		||\mathbf{u}-\mathbf{u}_h||_1 \leq C\inf_{\mathbf{u}^I\in\mathbf{V}_h,p^I\in Q_h}(||\mathbf{u}-\mathbf{u}^I||_1+||p-p^I||).
	\end{align}
	Here $C$ is a constant that depends on $\mathbf{u}$ and $p$.
	
\end{lemma}

\begin{proof}
	It is well-known that equation~(\ref{E:LBB-1}) has a unique solution $(\mathbf{u},p)\in \mathbf{V}\times Q$, due to the coercivity of $a_0(\cdot,\cdot)$ and the satisfaction of the B.B condition in $b(\cdot,\cdot)$.\\
	
	The solution $\mathbf{u}_h$ of equation~(\ref{E:LBB0}) is in subspace $\mathbf{V}_h^0=Ker(B_h)$, where 
	\begin{equation}
	\left<B_h \mathbf{u}_h,q_h \right>=b(\mathbf{u}_h,q_h),~\forall~\mathbf{w}_h \in \mathbf{V}_h,q_h\in Q_h.
	\end{equation}
	Thus by restricting the trial function $\mathbf{w}_h\in \mathbf{V}_h^0$,
	\begin{align}
		a_0(\mathbf{u}_h,\mathbf{w}_h) &= \left<\mathbf{f},\mathbf w_h\right>,~\forall\mathbf{w}_h \in \mathbf{V}_h^0.
	\end{align}
	Then there exists a unique $\mathbf{u}_h\in\mathbf{V}_h$ that satisfies equation~(\ref{E:LBB0}). Note $p_h$ may not have unique solution.\\
	
	Next we slightly tweak the discretized problem such that  space $Q_h$ is replaced by $\hat{Q}_h=Q_h/Ker(B_h^T)$, ,
	\begin{align}
		a_0(\mathbf{u}_h,\mathbf{w}_h)+b(\mathbf w_h,\hat{p}_h) &= \left<\mathbf{f},\mathbf w_h\right>,~\forall\mathbf{w}_h \in \mathbf{V}_h,\label{E:LBBtweak}\\
		\nonumber
		b(\mathbf u_h,\hat{q}_h) &= 0,~\forall \hat{q}_h\in \hat{Q}_h,
	\end{align}
	It is straight-forward to check that the B.B. condition is satisfied in equation~(\ref{E:LBBtweak}).
	Therefore, there exists a unique pair $(\mathbf{u}_h,\hat{p}_h)\in\mathbf{V}_h\times\hat{Q}_h$ as the solution of equation~(\ref{E:LBBtweak}). Apparantly $\mathbf{u}_h$ is also the unique solution of equation~(\ref{E:LBB0}).\\
	
	It is evident that the following relationship also holds,
	\begin{align}
		a_0(\mathbf{u},\mathbf{w}_h)+b(\mathbf w_h,p) &= \left<\mathbf{f},\mathbf w_h\right>,~\forall\mathbf{w}_h \in \mathbf{V}_h,\label{E:LBBtweak2}\\
		\nonumber
		b(\mathbf u,\hat{q}_h) &= 0,~\forall \hat{q}_h\in \hat{Q}_h,
	\end{align}
	Subtracting equation~(\ref{E:LBBtweak}) from equation~(\ref{E:LBBtweak2}), we have,
	\begin{align}
		a_0(\mathbf{u}_h-\mathbf{u}^I,\mathbf{w}_h)+b(\mathbf w_h,p_h-p^I) &= a_0(\mathbf{u}-\mathbf{u}^I,\mathbf{w}_h)+b(\mathbf w_h,p-p^I),~\forall\mathbf{w}_h \in \mathbf{V}_h,\label{E:LBBtweak2}\\
		\nonumber
		b(\mathbf u_h-\mathbf{u}^I,\hat{q}_h) &= b(\mathbf u-\mathbf{u}^I,\hat{q}_h),~\forall \hat{q}_h\in \hat{Q}_h.
	\end{align}
	Here $\mathbf{u}^I\in\mathbf{V}_h$, $p^I\in Q_h$ are any interpolations of $\mathbf{u}$ and $p$. Note $p^I$ is not in the truncated subspace $\hat{Q}_h$.
	Becuase the B.B. condition holds in equation~(\ref{E:LBBtweak2}),
	\begin{align}
		\exists \mathbf{u}_0 \in \mathbf{V}_h, \hat{\mathbf{u}} &\in  Ker(\hat{B}_h), \\ \nonumber
		\mathbf{u}_h-\mathbf{u}^I &= \mathbf{u}_0+\hat{\mathbf{u}},\label{E:LBBu0} \\ 
		\nonumber
		b(\mathbf u_0,\hat{q}_h) &= b(\mathbf u-\mathbf{u}^I,\hat{q}_h),~\forall \hat{q}_h\in \hat{Q}_h.
	\end{align}
	Here $\hat B_h$ is defined as,
	\begin{align}
		\nonumber \left< \hat B_h \mathbf{u}_h,\hat q_h \right>=b(\mathbf{u}_h,\hat q_h),~\forall~\mathbf{u}_h \in \mathbf{V}_h,\hat q_h\in \hat Q_h.
	\end{align}
	If the B.B condition does not hold in equation~(\ref{E:LBB0}), $Ker(B_h^T) \ne \Phi$, therefore any $p^I\in Q_h$ can be expressed as,
	\begin{align}
		p^I = \hat{p} + \tilde{p},~\hat p \in \hat Q_h,\tilde{p}\in Ker(B_h^T).
	\end{align}
	\\
	We now demonstrate that $Ker(\hat{B}_h)=Ker(B_h)$.
	
	First, if ${\mathbf{u}}_h \in Ker(B_h)$,
	\begin{align}
		\left< B_h\mathbf{u}_h,q_h\right>=0,~\forall q_h\in Q_h.
	\end{align}
	Because $\hat{Q}_h \subseteq Q_h$,
	it follows $Ker(B_h) \subseteq Ker(\hat{B}_h)$.
	
	Next, if ${\mathbf{u}}_h \in Ker(\hat B_h)$,
	\begin{align}
		\left< \hat B_h\mathbf{u}_h,\hat q_h\right>=0,~\forall \hat q_h\in \hat Q_h.
	\end{align}
	Pick any $q_h\in Q_h$, $q_h=\hat{q}_h+\tilde{q}_h$, $\hat{q}_h\in\hat Q_h,\tilde{q}_h\in Ker(B_h^T),$
	\begin{align}
		\left< B_h\mathbf{u}_h,q_h\right>=\left< B_h\mathbf{u}_h,\hat{q}_h+\tilde{q}_h\right>=0,~\forall q_h\in Q_h.
	\end{align}
	Then we have $Ker(\hat B_h) \subseteq Ker({B}_h)$.
	
	Therefore $Ker(\hat B_h) = Ker({B}_h)$.
	\\
	
	From the above discussion we have,
	\begin{align}                	
		&a_0(\hat{\mathbf{u}},\mathbf{w}_h)+b(\mathbf w_h,p_h-p^I) = \\
		\nonumber
		&a_0(\mathbf{u}-\mathbf{u}^I,\mathbf{w}_h)+b(\mathbf w_h,p-p^I)-a_0(\mathbf{u}_0,\mathbf{w}_h),~\forall\mathbf{w}_h \in \mathbf{V}_h,\mathbf{u}^I\in\mathbf{V}_h,p^I\in Q_h.
	\end{align}
	Let $\mathbf{w}_h=\hat{\mathbf{u}}$ and notice $b(\hat{\mathbf{u}},p_h-p^I)=0$, since $\hat{\mathbf{u}}\in Ker(\hat B_h)=Ker(B_h)$.
	It follows,
	\begin{align}
		a_0(\hat{\mathbf{u}},\hat{\mathbf{u}}) = a_0(\mathbf{u}-\mathbf{u}^I,\hat{\mathbf{u}})+b(\hat{\mathbf{u}},p-p^I)-a_0(\mathbf{u}_0,\hat{\mathbf{u}})
	\end{align}
	It is known that B.B. condition holds in equation~(\ref{E:LBBtweak}), thus we have
	\begin{align}
		\exists~\beta>0,~
		||\mathbf{u}_0||_1 \leq \beta^{-1}||\hat B_h\mathbf{u}_0||.
	\end{align}
	Therefore, from equation~(\ref{E:LBBu0}),
	\begin{align}
		||\mathbf{u}_0||_1 \leq \beta^{-1}\sup_{\hat q_h\in\hat Q_h}\frac{|b(\mathbf{u}_0,\hat q_h)|}{||\hat q_h||} = \beta^{-1}\sup_{\hat q_h\in\hat Q_h}\frac{|b(\mathbf{u}-\mathbf{u}^I,\hat q_h)|}{||\hat q_h||} \leq C\beta^{-1}||\mathbf{u}-\mathbf{u}^I||_1 \label{E:u0estimate}
	\end{align}
	
	Then,
	\begin{align}
		\alpha||\hat{\mathbf{u}}||_1^2 \leq a_0(\hat{\mathbf{u}},\hat{\mathbf{u}})
		\leq & ||a_0||\cdot||\mathbf{u}-\mathbf{u}^I||_1\cdot||\hat{\mathbf{u}}||_1+C||p-p^I||\cdot||\hat{\mathbf{u}}||_1+ \\ \nonumber &C\beta^{-1}||a_0||\cdot||\mathbf{u}-\mathbf{u}^I||_1\cdot||\hat{\mathbf{u}}||_1
	\end{align}
	
	\begin{align}
		\nonumber	\Rightarrow 
		||\hat{\mathbf{u}}||_1 \leq C(||\mathbf{u}-\mathbf{u}^I||_1+||p-p^I||)
	\end{align}
	Combining equation~(\ref{E:u0estimate}),
	\begin{align}
		||\mathbf{u}_h-\mathbf{u}_I||_1 \leq C(||\mathbf{u}-\mathbf{u}^I||_1+||p-p^I||).
	\end{align}
	Finally,
	\begin{align}
		||\mathbf{u}_h-\mathbf{u}||_1 \leq C\inf_{\mathbf{u^I\in\mathbf{V}_h},p^I\in Q_h}(||\mathbf{u}-\mathbf{u}^I||_1+||p-p^I||).
	\end{align}
	Here constant $C$ may be of different value in each equation for simplicity.
\end{proof}

\begin{cor} \label{C:interpolation}
	For any $\mathbf{u}\in\mathbf{V}^0$, where $\mathbf{V}^0$ is the divergence-free subspace of $\mathbf{V}$, there exists an interpolation $\tilde{\mathbf{u}}^I\in\mathbf{V}_h^0$ such that,
	\begin{align}
		||\tilde{\mathbf{u}}^I-\mathbf{u}||_1 \leq C\inf_{\mathbf{u^I\in\mathbf{V}_h},p^I\in Q_h}(||\mathbf{u}-\mathbf{u}^I||_1+||p-p^I||).
	\end{align} 
\end{cor}
This corollary shows that the divergence-free vector field $\mathbf{u}$ can be approximated by a weakly-divergence-free interpolation with the above error bounds. Whether or not the B.B condition is satisfied or not does not matter.

\begin{theorem}
	Given $\mathbf{f}\in\mathbf{V}$,
	there exists a unique pair $(\mathbf{u},p)\in \mathbf{V}\times Q$ such that,
	\begin{align}
		a(\mathbf{u},\mathbf{w}) + c(\mathbf{u},\mathbf{w})+b(\mathbf{w},p) &= \left<\mathbf{f},\mathbf{w} \right>,
		~\forall\mathbf{w} \in \mathbf{V}, \label{E:LBB1} \\
		\nonumber b(\mathbf{u},q) &= 0,~\forall q\in Q.
	\end{align}
	Here the binlinear form $b(\mathbf{u},q)$ satisfies the B.B. condition such that,
	\begin{align}
		\nonumber &\exists~ \beta > 0,~ 
		\inf_{q\in Q, q\ne 0} \sup_{\mathbf{u}\in\mathbf{V},\mathbf{u}\ne\mathbf{0}} \frac{|b(\mathbf{u},q)|}{||\mathbf{u}||_{\mathbf{V}} ||q||_Q} \geq \beta.
	\end{align}

	Then the following discretized linear system also has a unique solution $\mathbf{u}_h\in \mathbf{V}_h^0$,
	\begin{align}
		a_{supg}(\mathbf{u}_h,\mathbf{w}_h) + c_{supg}(\mathbf{u}_h,\mathbf{w}_h) =  \left<\mathbf{f},\mathbf{w}_h\right>+\left<\mathbf{f},\tau \mathbf{u}_0\cdot\nabla\mathbf{w}_h \right>_{\tilde{\Omega}}
		,~\forall\mathbf{w}_h \in \mathbf{V}_h. \label{E:LBB2}
	\end{align}
	Here the discretized binlinear form $b(\mathbf{u}_h,q_h)$ may not satisfy the B.B condition.
	
	If we assume $\mathbf{V}_h$ to be continous piecewise polynomials of degree $k$ and $Q_h$ be of degree $k-1$,
	the following esitmation holds,
	\begin{align}
		||\mathbf{u}-\mathbf{u}_h||_1^2 \leq C_{u,p}h^{2}.
	\end{align}
	Here $C_{u,p}$ is a function of $\mathbf{u}$ and $p$.
	
\end{theorem}

\begin{proof}
	
	It is well-known that when B.B. condition is satisfied, equation~(\ref{E:LBB1}) has a unqiue solution because the coercivity holds true.
	
	The unique solution $(\mathbf{u},p)\in \mathbf{V}\times Q$ also satisfies,
	\begin{align}
		a_{supg}(\mathbf{u},\mathbf{w}_h) + c_{supg}(\mathbf{u},\mathbf{w}_h)+b(\mathbf{w}_h,p) &+ \label{E:LBB3} \\
		\nonumber \left<\tau \mathbf{u}_0\cdot\nabla\mathbf{w}_h ,\nabla p\right>_{\tilde{\Omega}} &= \left<\mathbf{f},\mathbf{w}_h\right>+\left<\mathbf{f},\tau \mathbf{u}_0\cdot\nabla\mathbf{w}_h \right>_{\tilde{\Omega}}
		,~\forall\mathbf{w}_h \in \mathbf{V}_h, \\
		\nonumber b(\mathbf{u},q_h) &= 0,~\forall q_h\in Q_h.
	\end{align}
	
	Combining equation~(\ref{E:LBB2}) and~(\ref{E:LBB3}), 
	\begin{align}
		a_{supg}(\mathbf{u}_h,\mathbf{w}_h) + c_{supg}(\mathbf{u}_h,\mathbf{w}_h) &= \\
		\nonumber  a_{supg}(\mathbf{u},\mathbf{w}_h) + c_{supg}(\mathbf{u},\mathbf{w}_h) &+ b(\mathbf{w}_h,p) + \left<\tau \mathbf{u}_0\cdot\nabla\mathbf{w}_h ,\nabla p\right>_{\tilde{\Omega}}, ~\forall\mathbf{w}_h \in \mathbf{V}_h^0.
	\end{align}
	
	Recall $\mathbf{V}_h^0$ is the discrete divergence-free subspace.	
	Let $\mathbf{u}^I \in \mathbf{V}_h^0$, be any interpolation of $\mathbf{u}$, and $p^I \in Q_h$ be any interpolation of $p$. Note $b(\mathbf{w}_h,p^I)=0$. We have,
	\begin{align}
		a_{supg}(\mathbf{u}_h-\mathbf{u}^I,\mathbf{w}_h) &+ c_{supg}(\mathbf{u}_h-\mathbf{u}^I,\mathbf{w}_h) = 
		a_{supg}(\mathbf{u}-\mathbf{u}^I,\mathbf{w}_h) + c_{supg}(\mathbf{u}-\mathbf{u}^I,\mathbf{w}_h) \\ \nonumber &+ b(\mathbf{w}_h,p-p^I) + \left<\tau \mathbf{u}_0\cdot\nabla\mathbf{w}_h ,\nabla p\right>_{\tilde{\Omega}}, ~\forall\mathbf{w}_h \in \mathbf{V}_h^0.
	\end{align}
	
	Let $\mathbf{w}_h=\mathbf{u}_h-\mathbf{u}^I \in \mathbf{V}_h^0$, from lemma~\ref{L:coercivitySUPG},
	\begin{align}
		& \alpha_0||\mathbf{u}_h-\mathbf{u}^I||_1^2  + \frac{1}{6}||\tau^{1/2} \mathbf{u}_0\cdot\nabla(\mathbf{u}_h-\mathbf{u}^I)||_{\tilde{\Omega}}^2  \\
		\nonumber
		\leq &\left<\nu \nabla(\mathbf{u}-\mathbf{u}^I),\nabla(\mathbf{u}_h-\mathbf{u}^I)\right> +  \left< \mathbf{u}_0\cdot\nabla(\mathbf{u}-\mathbf{u}^I),\mathbf{u}_h-\mathbf{u}^I\right> \\ 
		\nonumber &+\left< (\mathbf{u}-\mathbf{u}^I)\cdot\nabla\mathbf{u}_{0},\mathbf{u}_h-\mathbf{u}^I\right>+\left< \gamma(\mathbf{u}-\mathbf{u}^I),\mathbf{u}_h-\mathbf{u}^I\right>\\
		\nonumber	&-\left< \nu\nabla^2(\mathbf{u}-\mathbf{u}^I),\tau \mathbf{u}_0\cdot\nabla(\mathbf{u}_h-\mathbf{u}^I)\right>_{\tilde{\Omega}} \\ \nonumber
		&+ \left< \mathbf{u}_0\cdot\nabla(\mathbf{u}-\mathbf{u}^I),\tau \mathbf{u}_0\cdot\nabla(\mathbf{u}_h-\mathbf{u}^I)\right>_{\tilde{\Omega}} \\
		\nonumber &+
		\left< (\mathbf{u}-\mathbf{u}^I)\cdot\nabla\mathbf{u}_{0},\tau \mathbf{u}_0\cdot\nabla(\mathbf{u}_h-\mathbf{u}^I)\right>_{\tilde{\Omega}} \\ \nonumber
		& +
		\left< \gamma(\mathbf{u}-\mathbf{u}^I),\tau \mathbf{u}_0\cdot\nabla(\mathbf{u}_h-\mathbf{u}^I)\right>_{\tilde{\Omega}}. \\
		\nonumber
		& + b(\mathbf{u}_h-\mathbf{u}^I,p-p^I) + \left<\tau \mathbf{u}_0\cdot\nabla(\mathbf{u}_h-\mathbf{u}^I) ,\nabla p\right>_{\tilde{\Omega}} \\ \nonumber
		\leq & \frac{\alpha_0}{2}||\mathbf{u}_h-\mathbf{u}^I||_1^2 + C_1||\mathbf{u}-\mathbf{u}^I||_1^2+\frac{1}{12}||\tau^{1/2} \mathbf{u}_0\cdot\nabla(\mathbf{u}_h-\mathbf{u}^I)||_{\tilde{\Omega}}^2 \\\nonumber
		&+ C_2\nu^2||\tau^{1/2}\nabla^2(\mathbf{u}-\mathbf{u}^I)||_{\tilde{\Omega}}^2 +
		C_3||p-p^I||^2 + C_4||\tau^{1/2}\nabla p||^2\\ \nonumber \\ \nonumber
		\Rightarrow~~~~ 
		&\frac{\alpha_0}{2}||\mathbf{u}_h-\mathbf{u}^I||_1^2 \\ \nonumber
		\leq & C_1||\mathbf{u}-\mathbf{u}^I||_1^2 + C_2\nu^2||\tau^{1/2}\nabla^2(\mathbf{u}-\mathbf{u}^I)||_{\tilde{\Omega}}^2+ C_3||p-p^I||^2 
		\\\nonumber
		& 
		+ C_4||\tau^{1/2}\nabla p||^2
	\end{align}
	Therefore,
	\begin{align}
		&~\exists~C>0,\\\nonumber
		&||\mathbf{u}-\mathbf{u}^I||_1^2 
		\leq  C(||\mathbf{u}-\mathbf{u}^I||_1^2 + \nu^2||\tau^{1/2}\nabla^2(\mathbf{u}-\mathbf{u}^I)||_{\tilde{\Omega}}^2+ ||p-p^I||^2+||\tau^{1/2}\nabla p||^2).
	\end{align}
	\\
	Since $\mathbf{u}\in \mathbf{V}^0$ and $\mathbf{u}^I\in\mathbf{V}^0_h$, from Corollary~\ref{C:interpolation},
	\begin{align}
		\exists~\mathbf{u}^I\in\mathbf{V}^0_h,~||\mathbf{u}-\mathbf{u}^I||_1 \leq C\inf_{\tilde{\mathbf{u}}^I\in \mathbf{V}_h,p^I\in Q_h}(||\mathbf{u}-\tilde{\mathbf{u}}^I||_1+||p-p^I||)\sim O(h^k),
	\end{align}
	
	From the choice of $\tau=\frac{h}{2\mathbf{u}_0}min(\frac{\alpha}{3},1)$,
	\begin{align}
		\tau &= O(\frac{h}{|\mathbf{u}_0|}),~Peclet~number~\alpha~is~large,\\\nonumber
		\tau &= O(\frac{h^2}{\nu}),~Peclet~number~\alpha~is~small.
	\end{align}
	Remember $\alpha=\frac{|\mathbf{u}_0|h}{2\nu}$ in each element.
	
	We therefore have,
	\begin{align}
		&\begin{cases}
			\nu^2\tau=\frac{\nu^2h}{2|\mathbf{u}_0|}=\frac{h^3|\mathbf{u}_0|}{8\alpha^2}\leq \frac{1}{8}h^3|\mathbf{u}_0|\leq Ch^3, ~\alpha~large,~|\mathbf{u}_0|~bounded,  \\ \nonumber
			\nu^2\tau=\frac{1}{12}\nu h^2\leq Ch^2,~\alpha~small,
		\end{cases}\\
		&\nu^2||\tau^{1/2}\nabla^2(\mathbf{u}-\mathbf{u}^I)||_{\tilde{\Omega}}^2 \leq C_{u}h^{2l},\\\nonumber
		&2l=\begin{cases}
			2k+1,~\alpha~large,\\
			2k,~\alpha~small,
		\end{cases}
	\end{align}
	where $C$'s are constant numbers in each estimation and $C_u$ is a function of $\mathbf{u}$. 

We also have,
\begin{align}
	&\begin{cases}
		\tau=\frac{h}{2|\mathbf{u}_0|}=\frac{h^2}{4\nu\alpha}\leq \frac{Ch^2}{\nu},~\alpha~large,\\
		\tau=\frac{h^2}{12\nu}\leq \frac{Ch^2}{\nu},~\alpha~small,
	\end{cases}\\\nonumber
	&||\tau^{1/2}\nabla p||^2 \leq C_p h^{2}.
\end{align}
Combining all results above and notice the leading order error is $O(h^2)$ when $k\geq 1$,
\begin{align}
	||\mathbf{u}-\mathbf{u}^I||_1^2 \leq C_{u,p} h^{2}.
\end{align}
Note all the $C_u$, $C_p$  are functions of $\mathbf{u}$ and $p$. Since given $\mathbf{f}$, the pair $(\mathbf{u},p)$ is unique. We can also write that,
\begin{align}
	||\mathbf{u}-\mathbf{u}^I||_1^2 \leq C_{f} h^{2}.
\end{align}

The leading-order error comes from $\left<\tau \mathbf{u}_0\cdot\nabla(\mathbf{u}_h-\mathbf{u}^I) ,\nabla p\right>_{\tilde{\Omega}}$. However, this scheme does not degrade the accuracy of $Q_1-P_0$ element.  It will be interesting to find a way to increase the order of convergence in the future.
\end{proof}

\begin{cor}
	Given the linear systems,
	\begin{align}
		a(T\mathbf{f},\mathbf{w}) + c(T\mathbf{f},\mathbf{w}) &= \left<\mathbf{f},\mathbf{w} \right>,
		~\forall\mathbf{w} \in \mathbf{V}^0, \label{E:LBB4} 
	\end{align}
	and,
	\begin{align}
		a_{supg}(T_h\mathbf{f},\mathbf{w}_h) + c_{supg}(T_h\mathbf{f},\mathbf{w}_h) &= \left<\mathbf{f},\mathbf{w}_h \right>+\left<\mathbf{f},\tau \mathbf{u}_0\cdot\nabla\mathbf{w}_h \right>_{\tilde{\Omega}},
		~\forall\mathbf{w}_h \in \mathbf{V}^0_h, \label{E:LBB5}
	\end{align}
	The operator $\mathbf{T}: \mathbf{V}\rightarrow\mathbf{V}$, $\mathbf{T}_h: \mathbf{V}\rightarrow\mathbf{V}_h$ is well-defined, with the norm bound,
	\begin{align}
		||(T-T_h)|_{\mathbf{E}(\mu)}||_{L(\mathbf{V})} \leq C h.
	\end{align}
	Here $T|_{\mathbf{E}(\mu)}$ means the restriction to the generalized eigenspace corresponding to eigenvalue $\mu$ of operator $T$. Here $\mu=\tilde\lambda^{-1}$, where $\tilde{\lambda}$ is defined in equation~(\ref{E:lambdaproblem}), and,
	\begin{align}
		\tilde{\lambda}T\mathbf{u} = \mathbf{u}.
	\end{align}
\end{cor}

The convergence of the flow eigenvalue problem is a direct result of Bab\v{u}ska and Osborn~\cite{babuskaosborn1991},
\begin{theorem}
	(Bab\v{u}ska and Osborn) The distance of eigenspace $\mathbf{E}(\mu)$ and the discretized eigenspace $\mathbf{E}_h(\mu)$ satisfies the following bound,
	\begin{align}
		\hat{\delta}(\mathbf{E(\mu)},\mathbf{E_h(\mu)}) \leq C||(T-T_h)|_{\mathbf{E}(\mu)}||_{L(\mathbf{V})}.
	\end{align}
	Here
	\begin{align}
		\hat{\delta}(\mathbf{E(\mu)},\mathbf{E_h(\mu)})&=max(\delta(\mathbf{E(\mu)},\mathbf{E_h(\mu)}),\delta(\mathbf{E_h(\mu)},\mathbf{E(\mu)})),\\ \nonumber
		\delta(\mathbf E,\mathbf{F}) &= \sup_{\mathbf{u}\in E, ||\mathbf{u}||_1=1} \inf_{\mathbf{v}\in F}||\mathbf{u}-\mathbf{v}||_1
	\end{align}
\end{theorem}

\begin{theorem}
	(Bab\v{u}ska and Osborn) Let $\mu$ be a non-zero eigenvalue of $T$ with algebraic multiplicity
	equal to $m$ and let $\hat \mu_h$ denote the arithmetic mean of the $m$ discrete eigenvalues of $T_h$ converging towards $\mu$. Let $\phi_1,\cdots, \phi_m$ be a basis of generalized
	eigenvectors in $E(\mu)$ and let $\phi^*_1,\cdots, \phi^*_m$ be a dual basis of generalized eigenvectors in $E^*(\mu)$. Then,
	\begin{align}
		|\mu-\hat \mu_h| \leq \frac{1}{m}\sum_{i=1}^{m}|\left<(T-T_h)\phi_i,\phi^*_i\right>|+C||T-T_h||_{L(\mathbf{V})}||T^*-T^*_h||_{L(\mathbf{V^*})}.
	\end{align}
\end{theorem}

\begin{theorem}
	The eigenvalue problem has an error estimate for the $i$th eigenvalue,
	\begin{align}
		|\tilde{\lambda}_i - (\frac{1}{m}\sum_{i_k=1}^{m}\tilde{\lambda}_{i_kh}^{-1})^{-1} | \leq C_i h^{2},
	\end{align}
	where $m$ is the algebraic multiplicity of $\tilde{\lambda}_i$ and $\tilde{\lambda}_{i_kh}$ are the $m$ discrete eigenvalues converging towards $\tilde{\lambda}$.
	
	The estimate for the $i$th eigenmode is,
	\begin{align}
		||\mathbf{u}^i-\mathbf{u}^i_h||_1 \leq C_i h.
	\end{align}
\end{theorem}
\begin{proof}
	The proof follows from Boffi~\cite{boffieig} Theorem $10.4$ and the above Bab\v{u}ska and Osborn theory.
\end{proof}

\end{document}